\documentclass[thm-restate,english]{lipics-v2019}
\usepackage{wrapfig}
\usepackage{nico}
\usepackage{microtype}
\nolinenumbers
\title{Dynamic network congestion games}
\author{Nathalie Bertrand}
       {Univ Rennes, Inria, CNRS, IRISA, France}{nathalie.bertrand@inria.fr}
       {https://orcid.org/0000-0002-9957-5394}{}
\author{Nicolas Markey}
       {Univ Rennes, Inria, CNRS, IRISA, France}{nicolas.markey@irisa.fr}
       {https://orcid.org/0000-0003-1977-7525}{}
\author{Suman Sadhukhan}
       {Univ Rennes, Inria, CNRS, IRISA, France}{suman.sadhukhan@inria.fr}
       {}{}
\author{Ocan Sankur}
       {Univ Rennes, Inria, CNRS, IRISA, France}{ocan.sankur@irisa.fr}
       {https://orcid.org/ 0000-0001-8146-4429}{}
\authorrunning{N.~Bertrand, N.~Markey, S.~Sadhukhan, O.~Sankur}


\Copyright{Nathalie~Bertrand, Nicolas~Markey, Suman~Sadhukhan, Ocan~Sankur}
\begin{document}
\maketitle

\begin{abstract}
Congestion games are a classical type of games studied in game theory,
in which $n$ players choose a resource, and their individual cost
increases with the number of other players choosing the same
resource. In~network congestion games (NCGs), the resources correspond
to simple paths in a graph, \emph{e.g.} representing routing options
from a source to a target. In~this paper, we~introduce a variant of
NCGs, referred to as \emph{dynamic NCGs}: in~this setting,
players take transitions synchronously,
they select their next transitions dynamically,
and 
they are charged a cost that depends on the number of players simultaneously using the same transition.
%(and suffer the associated cost) simultaneously,
%and dynamically choose their next step.
%Moreover, the
%cost of a player is affected by other players only if they
%simultaneously take the same transition.

\looseness=-1
We study, from a complexity perspective, standard concepts of game
theory in dynamic~NCGs: social optima, Nash equilibria, 
%price of stability and anarchy,
and subgame perfect equilibria.  Our~contributions are the following:
the~existence of a strategy profile with social cost bounded by a
constant is in \PSPACE and \NP-hard.
(Pure) Nash equilibria always exist in
dynamic NCGs; the~existence of a Nash equilibrium with bounded cost
can be decided in \EXPSPACE, and computing a witnessing strategy
profile can be done in doubly-exponential time.
%
%Computing a Nash equilibrium with
%bounded cost can be done in doubly-exponential time, and the
%associated decision problem is in \EXPSPACE. % One can
% compute in doubly-exponential time the price of stability and the
% price of anarchy. 
%Finally, deciding t
The~existence of a subgame perfect equilibrium with bounded cost can
be decided in 2\EXPSPACE, and a witnessing strategy profile can be
computed in triply-exponential time.

%\SS{Should we write complexities for symmetric and asymmetric cases separately?}

\end{abstract}

\section{Introduction}
\label{sec:intro}
Congestion games model selfish resource sharing among several
players~\cite{R-IJGT73}. A special case is the one of network
congestion games,
%or selfish routing games,
in which players aim at
routing traffic through a congested network. Their popularity is
certainly due to the fact that they have important practical
applications, whether in transportation networks, or in large
communication networks% such as the Internet
~\cite{QYZS-ieeeTN06}. In
network congestion games, each player chooses a set of transitions,
forming a simple path from a source state to a target state, and the
cost of a transition increases with its load, that is, with the number
of players using it.
%Players are self-interested and aim at minimizing their costs.

Network congestion games can be classified into atomic and non-atomic
variants. Non-atomic semantics is appropriate for large populations of
players, thus seen as a continuum. One then considers portions of the
population that apply predefined strategies, and there is no such
thing as the effect of an individual player on the cost of others.  In
contrast, in atomic games, the number of players is fixed, and each
player may significantly impact the cost other players incur.
We only focus on atomic games in this paper.

\subparagraph*{Network congestion games.}
Network
congestion games, also called atomic selfish routing games in the
literature, were first considered by Rosenthal~\cite{R-IJGT73}.
These games are defined by a directed graph, a~number of pairs of
source-target vertices, and non-decreasing cost functions for each
edge in the graph. For each source-target pair, a player must choose a
route from the source to the target vertex. Given their choice of
simple paths, the cost incurred by a player depends on the number of
other players that choose paths sharing edges with their path,
and on the cost functions of these edges. In this setting, a Nash
equilibrium maps each player to a path in such a way that no player has an
incentive to deviate: they cannot decrease their cost by choosing a
different path.

Rosenthal
proved that they are potential games, so that Nash equilibria always
exist. Monderer and Shapley~\cite{MS-geb96} studied in a more general way
potential games, and explained how to iteratively use best-response
strategies to converge to an equilibrium.
Interestingly, under reasonable assumptions on the cost functions,
Bertsekas and Tsitsiklis established that there is a direct
correspondence between equilibria in selfish routing and distributed
shortest-path routing, which is used in practice for packet routing in
computer networks~\cite{BT-book89}.
We refer the interested reader to~\cite{roughgarden-chap2007}
for an introduction and many basic results on general routing games.

A natural question is whether selfish routing is very different from a
routing strategy decided by a centralized authority. In~other words,
how far can a selfish optimum be from the social optimum, in which
players would cooperate. The notion of price of anarchy, first
proposed by Koutsoupias and Papadimitriou~\cite{KP-csr09},
%\NM{Older paper (1999)?},
is the ratio of the worst cost of a Nash equilibrium
and the cost of the social optimum.
%The~closer to~$1$,
%the better the Nash equilibrium, compared to the
%collaborative social optimum.
This measures how bad Nash equilibria can~be.
In the context of network congestion games, the price of anarchy was
first studied by Suri~\textit{et~al.}~\cite{STZ-Algo07}, establishing
an upper bound of $\frac 5 2$ when all cost functions are affine.
A~refined upper bound was provided by Awerbuch
\textit{et~al.}~\cite{AAE-stoc05}. Bounds on the dual notion of price
of stability, which is the ratio of the cost of a best Nash
equilibrium and the cost of the social optimum was also studied for
routing games~\cite{ADKTWR-jc08}.

\subparagraph*{Timing aspects.}
Several works investigated refinements of this setting.
In \cite{hmrt-tcs11}, the authors study network congestion games
in which each edge is traversed with a fixed duration independent of its load,
while the cost of each edge depends on the load.
The model is thus said to have \emph{time-dependent} costs since the load depends on the times
at which players traverse a given edge. The authors prove the existence of
Nash equilibria by reduction to the setting of \cite{R-IJGT73}.
An extension of this setting with
%timed automata and clocks
timed constraints
was studied in \cite{AGK-mfcs17,AGK-mfcs18}.
%\SS*{\cite{AGK-mfcs17, AGK-mfcs18} do not necessarily use timed automata, they define priced Timed automata to reduce existence of NE problem, and that's all}
%\NM*{at least [3] does: in \url{https://arxiv.org/abs/1808.04882}, section 2.2, a \emph{timed network}
%  is a timed autmoaton. So I think we can keep this sentence like this.}
%Here, the games are defined using timed automata, so rich timing constraints
%can be imposed in the graph.

The setting of fixed durations with time-dependent costs is
interesting in applications where the players sharing a resource (an
edge) see their quality of services decrease, while the time to use
the resource is unaffected~\cite{AGK-mfcs18}.
Timing also appears, for instance, in \cite{ppt-mor2009,kp-icalp12}
where the load affects travel times and players' objective is to
minimize the total travel time.  Other works focus on flow models with
a timing aspect~\cite{koch2011nash,bfa-geb2015}.

% More recently, Avni et al. considered congestion games in which
% players choose multisets of transitions (instead of sets of
% transitions)~\cite{AKT-fsttcs15}, and also dynamic resource allocation
% games~\cite{AHK-tcs20}. Our work is part of this research direction,
% at the crossroad of formal methods and algorithmic game theory.
% \begin{itemize}
% \item Atomic selfish routing was first considered by
%   Rosenthal~\cite{R-IJGT73}, who proved that routing games are
%   potential games, so that Nash equilibria always exist. He also
%   coined the termilogy of \emph{congestion games} in which the
%   identity of players do not impact their costs.
% \item Monderer and Shapley~\cite{MS-geb96} studied more generally
%   potential games, and explained how to iteratively use best-response
%   strategies to converge towards an equilibrium.
% \item The price of anarchy was first studied by Suri et
%   al.~\cite{STZ-Algo07}, establishing an upper bound of $\frac 5 2$
%   when all cost functions are affine. A refined upper bound was then
%   provided by Awerbuch et al.~\cite{AAE-stoc05}.
% \item Price of stability in atomic selfish routing
%   games, for instance see~\cite{ADKTWR-jc08}.
%   %+ Caragiannis et al. ICALP 2006
%   %+ Christodoulou/Koutsoupias ESA'05.
% \end{itemize}

\subparagraph*{Dynamic network congestion games.} 
In classical network congestion games, including those mentioned above, 
players choose their strategies (\textit{i.e.,} their \emph{simple} paths) in one shot.
However, it may be interesting to let agents choose their paths \emph{dynamically},
that is, step by step, by observing
other players' previous choices. In this paper, 
we study network congestion games with time-dependent costs
as in \cite{hmrt-tcs11}, but with unit delays, and in a dynamic setting.
More~precisely,
at~each step, each of the players simultaneously selects the edge
they want to take;
%an edge is simultaneously chosen by each player at each step;
each player is then charged a cost that depends on the load of the
edge they selected, and traverses that edge in one step. We~name
these games \emph{dynamic network congestion games} (dynamic
NCGs~in~short); the~behaviour of the players in such games is formalized
by means of \emph{strategies}, telling the players what to play
depending of the current configuration.
Notice that, because congestion effect applies to edges used
\emph{simultaneously} by several players, taking cycles can be interesting
in dynamic NCGs, which makes our setting more complex than most NCG
models~\cite{AHK-tcs20, hmrt-tcs11, R-IJGT73, roughgarden-chap2007}.
%As a player's cost incur the congestion effect for an edge only when they select that edge with other player \emph{simultaneously}, the dynamically chosen paths may include cycles/loops unlike in most of the classical NCG models.

Such a dynamic setting was studied in~\cite{AHK-tcs20} for resource
allocation games, which extends~\cite{R-IJGT73} with dynamic choices.
% We give a more detailed comparison of our setting with related works
% at the end of this section.
A more detailed related work appears 
at the end of this section.

\subparagraph*{Standard solution concepts.}
% \leavevmode\NM{please check is this paragraph is ok}
We study classical solution concepts on dynamic network congestion
games.  A~strategy profile (\emph{i.e.},~a~function assigning a
strategy to each player) is a \emph{Nash Equilibrium}~(NE) when each
single strategy is an optimal response to the strategies of the other
players; in~other terms, under such a strategy profile, no~player may
lower their costs by unilaterally changing their strategies.  Notice
that NEs need not exist in general, and when they exist, they may not
be unique.  In~the setting of dynamic games, Nash~Equilibria are
usually enforced using \emph{punishing strategies}, by~which any
deviating player will be punished by all other players once the
deviation has been detected. However, such punishing strategies may also
increase the cost incurred to the punishing players, and hence do not
form a credible threat; \emph{Subgame-Perfect Equilibria}~(SPEs)
refine~NEs and
%correcting this aspect.
address this issue by requiring that the strategy profile is an NE 
along any~play.%\NM{Clear enough?}

NEs and SPEs
%correspond to selfish behaviours
aim at minimizing the individual cost of each player (without caring
of the others' costs); in~a collaborative setting, the~players may
instead try to lower the social cost, i.e.,~the~sum of the costs
incurred to all the players. Strategy profiles achieving this are
called \emph{social optima}~(SO). Obviously, the social cost of NEs
and SPEs cannot be less than that of the social optimum;
the~\emph{price of anarchy} measures how bad selfish
behaviours may be compared to collaborative ones.

%% On~the previous example, $\pi_1$~and $\pi_2$ were quite
%% simple strategies, called \emph{blind}; they~consist of paths that
%% players decide before starting the game. The~tuple $(\pi_1,\pi_2)$
%% forms a strategy profile, and one can show that it is a \emph{Nash
%%   equilibrium} (among blind strategies): neither Player~1 nor Player~2
%% %has an incentive to
%% can choose another path that would decrease their cost, while their
%% adversary keeps the same strategy. It~is even a \emph{social optimum}:
%% the~total cost of this profile is~$12$, and one can easily check that
%% it is the minimum cost (among all strategies) when two players are
%% involved. Of~course, not all Nash equilibria are social optima, and
%% it~may happen that the cost of social optima is less than the cost of
%% any Nash equilibrium. To~quantify this gap, two measures have been
%% introduced: \emph{price of anarchy} and \emph{price of
%%   stability}. They consist of the ratio of the cost of a worst
%% (resp. best) Nash equilibrium and the cost of the social
%% optimum. In~the example of Figure~\ref{fig:dynamicVSstatic}, the~price
%% of stability is $1$.

\subparagraph*{Our contributions.} 
We take a computational-complexity viewpoint to study dynamic network
congestion games. We first establish the complexity of computing the
social optimum, which we show is in~\PSPACE and \NP-hard.
% (Section~\ref{section:socopt}).
We~then prove that best-response strategies can be computed in
polynomial time, and that dynamic NCGs are potential games, thereby
showing the existence of Nash equilibria in any dynamic NCG;
% (Section~\ref{section:nash-existence});
this~also shows that some Nash equilibrium can be computed in
pseudo-polynomial time.  We~then give an \EXPSPACE (resp. 2\EXPSPACE)
algorithm to decide the existence of Nash Equilibria
(resp. Subgame-Perfect Equilibria) whose costs satisfy given bounds.
This
%with bounded cost in~\EXPSPACE. This~
allows us to compute best and worst such equilibria,
and then the price of anarchy and the price of stability.
%\NM{Also mention SPEs}

Note that some of the high complexities follow from the binary
encoding of the number of players, which is the main input parameter.
% results are due to the fact that
% the main input parameter is the number of players, which is given in
% binary.
For~instance, the exponential-space complexity drops to
pseudo-polynomial time for a fixed number of players.  This~parameter
becomes important since we advocate the study of computational
problems, such as computing Nash equilibria with a given cost
bound. We also believe that computing precise values for price of
anarchy and the price of stability is interesting, rather than
providing bounds on the set of all instances as in
\textit{e.g.}~\cite{STZ-Algo07}.
         
        %We take a computational-complexity viewpoint to study dynamic
        %network congestion games. In Section~\ref{section:socopt} we
        %establish the complexity of deciding whether the cost of a
        %social optimum is less than a given bound. In
        %Section~\ref{section:nash}, we settle how~costly it is to
        %compute a Nash equilibrium, or to compute the price of anarchy
        %and the price of stability.  Such questions depart from the
        %standard goals in algorithmic game theory of
        %\emph{e.g.}~establishing lower and upper bounds on the price
        %of anarchy and on the price of stability for a whole class of
        %instances.

\subparagraph*{Comparison with related work.}
The works closest to our setting
are~\cite{hmrt-tcs11,AHK-tcs20,AGK-mfcs17,AGK-mfcs18}.
As~in~\cite{hmrt-tcs11,AGK-mfcs18}, we~establish the existence of Nash
equilibria using potential games. Unlike~\cite{hmrt-tcs11}, we~cannot
obtain this result immediately by reducing our games to congestion
games~\cite{R-IJGT73} since the lengths of the strategies cannot be
bounded \emph{a~priori}. Moreover, the best-response problem has a
polynomial-time solution in our setting
% (see~Lemma~\ref{thm:br-algorithm})
while the problem is \NP-hard both
in~\cite{hmrt-tcs11,AGK-mfcs18}. In~\cite{hmrt-tcs11}, this is due to
the possibility of having arbitrary durations, while the source of
complexity in \cite{AGK-mfcs17,AGK-mfcs18} is due to the use of timed
automata.  Thus, our setting offers a simpler way of expressing
timings, and avoids their high complexity for this problem.
        
        Dynamic choices were studied in~\cite{AHK-tcs20} but with a different cost model. Moreover, network congestion games can only be reduced to such a setting given an a priori  bound on the length of the paths.
        %If a bound on the lengths of the strategies is given, then our games could be encoded using
        %dynamic resource allocation games of~\cite{AHK-tcs20}. However,
        So we cannot directly transfer any of their results to our setting.
        Dynamic choices were also studied in \cite{hmrt-tcs11}
        in the setting of coordination mechanisms which are local policies that allow one to sequentialize traffic on the edges.

\section{Preliminaries}
\label{section:prelim}
\subsection{Dynamic network congestion games}

%We consider network congestion games (NCG).
\begin{wrapfigure}[10]{r}{.44\textwidth}
  \centering
  \begin{tikzpicture}[xscale=1]
    \path[use as bounding box] (0,.8) -- (5,-1.3);
    \draw (0,0) node[rond5] (s) {} node {$\vphantom{tg}\src$};
    \draw (1.5,1.2) node[rond5] (1) {$\vphantom{tg}v_1$};
    \draw (1.5,-1.2) node[rond5] (2) {$\vphantom{tg}v_2$};
    \draw (3,0) node[rond5] (3) {} node {$v_3$};
    \draw (4.7,0) node[rond5] (t) {} node {$\tgt$};
    \draw[-latex'] (s) -- (1) node[sloped,above,midway] {$x\mapsto x$};
    \draw[-latex'] (s) -- (2) node[sloped,below,midway] {$x\mapsto 5$};
    \draw[-latex'] (1) -- (2) node[above,midway,fill=white] {$x\mapsto 6$};
    \draw[-latex'] (1) -- (3) node[sloped,above,midway] {$x\mapsto 3x$};
    \draw[-latex'] (2) -- (3) node[sloped,below,midway] {$x\mapsto x$};
    \draw[-latex'] (3) -- (t) node[sloped,below,midway] {$x\mapsto 4x$};
  \end{tikzpicture}
  \caption{Representation of an arena for a dynamic NCG (loop omitted on~\tgt)}\label{fig-ex}
\end{wrapfigure}
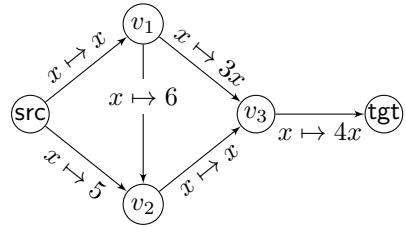
Let~$\calF$ be the family of non-decreasing functions from~$\bbN$
to~$\bbN$
that are piecewise-affine, with finitely many
pieces. We~assume that each $\edgecost \in \calF$ is represented by the endpoints of intervals, and the coefficients, all encoded in
binary.
An~arena for dynamic network congestion games is a weighted
graph $\calA=\tuple{V,E,\src,\tgt}$, where $V$ is a finite set of
states,
%$E\subseteq V\times \calF\times V$ is a set of transitions
%labelled with a nondecreasing cost function,
${E\colon V\times V \to \calF}$ is a partial function defining the
cost of edges, and \src and \tgt are a
source- and a target state in~$V$. It~is assumed throughout this paper
that~$\tgt$ has only a single outgoing transition, which is a
self-loop with constant cost function~$x\mapsto 0$. We~also assume
that~$\tgt$ is reachable from all other states.
%\NM*{We try to allow cost=0. Check at the end that it is ok. If not, explain that we still have 0-loop on \tgt for convenience.}

A dynamic network congestion game (dynamic~NCG for~short) is a pair
$\calG=\tuple{\calA,n}$ where $\calA$ is an arena as above and
$n\in\bbN$ is the (binary-encoded) number of players.  In~a dynamic
network congestion game, all~players start from~$\src$ and
simultaneously select the edges they want to traverse, with the aim of
reaching the target state with minimal individual accumulated cost.
Taking an edge~$e=(v,\edgecost,v')$ has a cost $\edgecost(l)$, where
$l$ is the number of players simultaneously using edge~$e$.
%\SS{To me, it is not clear, whether $l$ is the number of player who are using $e$ at the same step}
%
The~cost function of edge~$e$ is
denoted by~$\edgecost_e$. We~let $\kappa=\max_{e\in E} \edgecost_e(n)$,
which is the maximal cost that a player may endure along one edge.
%
%Network congestion games are concurrent games, where all players
%simultaneously select the edges they want to traverse, with the aim of
%reaching the target state with minimal individual accumulated cost.

Our setting differs from classical network congestion
games~\cite{roughgarden-chap2007} mainly in two respects:
\begin{itemize}
\item first, the game is played in rounds, during which all players
  take exactly one transition;   the~number of players using an edge~$e$
  is measured \emph{dynamically}, at each round;

\item second, during the play, players may adapt
  their choices to what the other players have been doing in the previous
  rounds. 
\end{itemize}

%We~begin with defining two weighted graphs representing two different
%views of the behaviour of our dynamic network congestion games, and
%then define their semantics as a concurrent game.

\begin{remark}
  In this work, we mainly focus on the \emph{symmetric} case, where
  all players have the same source and target. This is because we take
  a parametric-verification point of view, with the (long-term) aim of checking
  properties of dynamic~NCGs for arbitrarily many players.
  An~important consequence of this choice is that the number of players
  now is encoded in binary, which results in an exponential blow-up
  in the number of configurations of the game
  (compared to the asymmetric setting).
\end{remark}

\subparagraph{Semantics as a concurrent game.}  For any~$k\in\bbN$, we
write $\set k=\{i\in\bbN \mid 1\leq i\leq k\}$. A~\emph{configuration}
of a dynamic network congestion game $\tuple{\calA,n}$ is a mapping
${c\colon \set n \to V}$, indicating the position of each player in
the arena.  We~define $c_\src\colon i\in\set n\mapsto \src$ and
$c_\tgt\colon i\in\set n \mapsto \tgt$ as the initial and target
configurations, respectively.

With $\tuple{\calA,n}$, we~first associate a multi-weighted
graph~$\calM=\tuple{C,T}$, where $C=V^{\set n}$
%\SS{Shouldn't it be $V^n$?}
is the set of all
configurations and $T\subseteq C\times \bbN^{\set n} \times C$ is a
set of edges, defined as follows: there is an edge $(c,w,c')$ in~$T$
if, and only~if, there exists a collection
$\mathbf e = (e_i)_{i\in\set n}$ of edges of~$E$ such that for
all~$i\in\set n$, writing $e_i=(v_i,\edgecost_i,v'_i)$ and
$u_i=\#\{j\in\set n\mid e_j = e_i\}$, we have $c(i)=v_i$,
$c'(i)=v'_i$, and $w(i)=\edgecost_i(u_i)$. We~denote this edge with
$c\xRightarrow{\mathbf e} c'$. We~may omit to mention~$\mathbf e$ since
it can be obtained from~$c$ and~$c'$;
similarly, we~write $\cost_i(c,c')$ for~$w(i)$.

Two~edges~$(c,w,c')$ and $(d,x,d')$, in that order, are said to be
\emph{consecutive} whenever $d=c'$.  Given a configuration~$c$,
a~\emph{path} from~$c$ in a dynamic network congestion game is either
the single configuration~$c$ (we~call this a trivial path) or a
non-empty, finite or infinite sequence of consecutive edges
$\rho=(t_j)_{1\leq j < |\rho|}$ in~$\calM$, where~$t_1$ is a
transition from~$c$; the size of a path~$\rho$ is one for trivial
paths, and $|\rho|\in\bbN\cup\{+\infty\}$ otherwise.
We~write $\Paths(\tuple{\calA,n},c)$ and $\Paths^\omega(\tuple{\calA,n},c)$ for
the set of finite and infinite paths from~$c$ in~$\tuple{\calA,n}$,
respectively.

With~each path~$\rho=(c_j,w_j,c'_j)_j$, and each player~$i\in\set n$,
we~associate a \emph{cost}, written $\cost_i(\rho)$, which is~zero for
trivial paths, $+\infty$ for infinite paths
along which $c_j(i)\not=\tgt$ for all~$j$,
and $\sum_{j=1}^{|\rho|-1} w_j(i)$ otherwise.
We~define the~\emph{social cost} of~$\rho$, denoted
by~$\soccost(\rho)$, as $\sum_{i\in\set n} \cost_i(\rho)$.

Given a path~$\rho$, an index~$1 \leq j < |\rho|+1$ and a
player~$i\in\set n$, we~write~$\rho(j)$ for the $j$-th configuration
of~$\rho$, and $\rho(j)(i)$ for the state of Player~$i$ in that
configuration.
%For~example, if~$\rho=(c_1,f_1,c_2)(c_2,f_2,c_3)$,
%then~$|\rho| = 3$, and $\rho(1) = c_1$, $\rho(2)=c_2$, and
%$\rho(3)=c_3$.
For~$j\geq 2$, we define~$\rho_{\leq j}$ as the prefix of~$\rho$ that
ends in the $j$-th configuration;
%\textit{e.g.}~${\rho_{\leq 2} = (c_1,f_1,c_2)}$.
we~let $\rho_{\leq 1}=\rho(1)$.
Similarly, for~$1\leq j\leq |\rho|-1$, we~let~$\rho_{\geq j}$ denote
the suffix that starts at the $j$-th configuration.
%\textit{e.g.}~${\rho_{\geq 2} = (c_2,f_2,c_3)}$.
Finally, if~$|\rho|$ is finite, we~let $\rho_{\geq|\rho|}=\rho(|\rho|)$.

\begin{example}\label{ex-ex1}
  Consider the arena~$\calA$ displayed at Fig.~\ref{fig-ex} and the
  dynamic NCG~$\tuple{\calA,2}$ with two players.  Assume that
  Player~$1$ follows the
  path~$\pi_1\colon \src \to v_1 \to v_3 \to \tgt$ and Player~$2$ goes
  via $\pi_2\colon \src \to v_1 \to v_2 \to v_3\to \tgt$. This gives
  rise to the following path:
\[
\left(\begin{array}{c} 1\mapsto\src\\2\mapsto\src
\end{array}\right)
\xrightarrow{\substack{1\mapsto 2\\2\mapsto 2}}
\left(\begin{array}{c} 1\mapsto v_1\\2\mapsto v_1
\end{array}\right)
\xrightarrow{\substack{1\mapsto 3\\2\mapsto 6}}
\left(\begin{array}{c} 1\mapsto v_3\\2\mapsto v_2
\end{array}\right)
\xrightarrow{\substack{1\mapsto 4\\2\mapsto 1}}
\left(\begin{array}{c} 1\mapsto\tgt\\2\mapsto v_3
\end{array}\right)
\xrightarrow{\substack{1\mapsto 0\\2\mapsto 4}}
\left(\begin{array}{c} 1\mapsto\tgt\\2\mapsto \tgt
\end{array}\right)
\]
Notice how edge~$v_3\to\tgt$ of~$\calA$ is used by both players,
but not simultaneously, so~that the cost of using that edge
is~$4$ for each of them, while it would be 8 in classical NCGs.
\end{example}

We~now extend this graph to a concurrent game structure.
A~\emph{move} for Player~$i\in \set n$ from configuration~$c$ is an
edge~$e=(v,\edgecost,v')\in E$ such that $v=c(i)$. A~\emph{move
  vector} from~$c$ is a sequence ${\mathbf e=(e_i)_{i\in\set n}}$ such
that for all~$i\in\set n$, $e_i$~is a move for Player~$i$ from~$c$.

A~network congestion game~$\tuple{\calA,n}$ then gives rise to a
concurrent game structure $\calS=\tuple{C,T,M,U}$ where $\tuple{C,T}$
is the graph defined above, $M\colon C\times \set n \to 2^E$ lists the
set of possible moves for each player in each configuration, and
$U\colon C\times E^{\set n} \to T$ is the transition function, such
that for every configuration~$c$ and every move vector~$\mathbf
e=(e_i)_{i\in\set n}$ with $e_i\in M(c,i)$ for all~$i\in\set n$,
$U(c,e)= ( c\xRightarrow{\mathbf e} c' )$.

A~\emph{strategy} for Player~$i$ in~$\calS$ from configuration~$c$ is
a function~$\sigma_i\colon \Paths(\tuple{\calA,n},c) \rightarrow E$
that associates, with any finite path~$\rho$ from~$c$ in~$\calS$,
a~move for this player from the last configuration of~$\rho$.
A~\emph{strategy profile} is a family~$\sigma=(\sigma_i)_{i\in\set n}$
of strategies, one for each player. We~write $\frakS$ for the set of
strategies, and $\frakS^n$ for the set of strategy profiles.

Let~$c$ be a configuration, $h$~be a finite path from~$c$ and a
strategy profile~$\sigma=(\sigma_i)_{i\in\set n}$
from~$c$. The~\emph{residual strategy profile} of~$\sigma$ after~$h$
is the strategy profile~$\sigma^h=(\sigma^h_i)_{i\in\set n}$ from the
last configuration of~$h$ defined by
$\sigma^h_i(h')=\sigma_i(h\cdot h')$, where $h\cdot h'$ is the
concatenation of paths~$h$ and~$h'$.

The~\emph{outcome} of a strategy profile~$\sigma$ from~$c$ is the
infinite path~$\rho=(c_i,w_i,c_{i+1})_{i\geq 1}$, hereafter
denoted with~$\outcome(\sigma)$, obtained by running the strategy
profile; it~is formally defined as the only infinite path such that
$(c_1,w_1,c_2) = U(c,\sigma(c))$, and such that for any~${j\geq 2}$,
$(c_j,w_j,c_{j+1}) = U(c_j,\sigma(h'))$, where
$h'= (c_1,w_1,c_2) \cdots (c_{j-1},w_{j-1},c_j)$.
%We~denote~it with~$\outcome(\sigma)$.

Pick a strategy profile~$\sigma=(\sigma_i)_{i\in\set n}$, and let
$\rho=(t_j)_{j\geq 1}$ be its outcome, writing
$t_j=(c_j,(w_j^i)_{i\in\set n}, c'_j)$ for all~$j\geq 1$.
Let~$k\in\set n$.  If~$c'_l(k)=\tgt$ for some~$l\in\bbN$, then
$\sigma_k$ is said to be winning for Player~$k$.  In~that case,
we~define \( \cost_k(\sigma) \) as $\cost_k(\outcome(\sigma))$.
If~$c'_l(i)=\tgt$ for all~$i\in\set n$, we~define the~\emph{social
  cost} of~$\sigma$ as
$\soccost(\sigma)=\soccost(\rho)$.

A strategy~$\sigma_i$ for Player~$i$ is said \emph{blind} whenever
%the following holds:
for any two finite paths~$\rho$ and~$\rho'$ having same length~$k$, if
for any position~$0\leq j<k$ we have $\rho(j)(i)=\rho'(j)(i)$, then
$\sigma_i(\rho)=\sigma_i(\rho')$. Intuitively, this means that
strategy~$\sigma_i$ follows a path in~$\calA$, independently of what
the other players~do. A blind strategy can thus be represented as
a path and we~write $|\sigma_i|$ for the length of that path (until its
first visit to~$\tgt$, if~any).  We~write $\frakB$ for the set of
blind strategies.

\begin{example}\label{ex-ex2}
  Consider again the arena~$\calA$ of
  Fig.~\ref{fig-ex}. The~paths~$\pi_1$ and $\pi_2$ from
  Example~\ref{ex-ex1} are two blind strategies in that
  dynamic~NCG. In~a 2-player setting, an~example of a non-blind
  strategy~$\sigma$ consists in first taking the transition
  $\src\to v_1$, and then either taking $v_1\to v_3$ if the other
  player took the same initial transition, or taking $v_1 \to v_2$
  otherwise.
\end{example}

\subparagraph{Representation as a weighted graph.}
%Since all players have symmetric roles, a
Another way of representing configurations is to consider their Parikh
images.  With a configuration~$c\in V^{\set n}$,
we~associate an abstract configuration $\bar c\in \set n^V$ defined as
$\bar c(v)=\#\{i\in\set n\mid c(i)=v\}$.
%Notice that for any configuration~$c$ of~$\tuple{\calA,n}$, 
%we~have $\sum_{v\in V} \bar c(v) =n$.

The~\emph{abstract weighted graph} associated with a dynamic network
congestion game~$\tuple{\calA,n}$ is the weighted graph
$\calP=\tuple{A,B}$, where $A$ contains all abstract configurations,
and there is an edge $(a,w,a')$ in~$B\subseteq A\times\bbN\times A$
if, and only~if, there is a mapping $b\colon E\to\set n$ such that
%the following conditions hold:
$\sum_{e\in E} b(e)=n$ and for all~$v\in V$, 
\begin{xalignat*}3
  a(v) &= \sum_{e=(v,\edgecost,v')} b(e) &
  w &= \sum_{e=(v,\edgecost,v')} b(e)\times \edgecost(b(e)) &
  %\text{ for all }v\in V \\
  a'(v) &=  \sum_{e=(v',\edgecost,v)} b(e).
%  & \text{ for all }v\in V 
\end{xalignat*}
\looseness=-1 Similarly to the representation as multi-weighted
graphs, an~\emph{abstract path} of a network congestion game is either
a single configuration or a non-empty, finite or infinite sequence of
consecutive edges in the abstract weighted graph.
The~\emph{cost} of an abstract path is the sum of the weights of its
edges (if~any).
%There is a direct correspondence between costs in
%$\calM$ and in $\calP$:
Then:
\begin{lemma}\label{lemma-linkMP}
  For any~$w\in\bbN\cup\{+\infty\}$, there is an abstract path in~$\calM$ with social cost~$w$ if, and only~if, there is an abstract path in~$\calP$ with cost~$w$.
\end{lemma}

\subsection{Social optima and equilibria}

Consider a dynamic network congestion
game~$\calG=\tuple{\calA,n}$. We~recall standard ways of optimizing
the strategies of the players, depending on the situation.

In a collaborative situation, all players want to collectively
minimize the total cost for having all of them reach the target state
of the arena. Formally, a strategy profile~$\sigma=(\sigma_i)_{i\in\set n}$
realizes the \emph{social optimum} if
$\soccost(\sigma) = \inf_{\tau\in\frakS^n} \soccost(\tau)$.

In a selfish situation, each player aims at optimizing their response to
the others' strategies.
%For~a~strategy
%profile~$\sigma=(\sigma_i)_{i\in\set n}$, a player~$k\in\set n$ and a
%strategy~$\sigma'_k$, we~write $\sigma[k\to \sigma'_k]$ for the
%strategy profile derived from~$\sigma$ by changing Player~$k$'s
%strategy to~$\sigma'_k$.
%
Given a strategy profile~$\sigma=(\sigma_i)_{i \in \set{n}}$, 
a player~$k \in \set{n}$,  and a~strategy~$\sigma_k' \in \frakS$,
we~denote by
$\tuple{\sigma_{-k},\sigma_k'}$ the strategy profile
$(\tau_i)_{i\in\set n}$ such that ${\tau_k=\sigma'_k}$ and
${\tau_i=\sigma_i}$ for all~${i\in \set n\setminus\{k\}}$.
The~strategy $\sigma_k$ is a \emph{best response} to
$(\sigma_i)_{i\in\set n\setminus\{k\}}$ if
$\cost_k(\sigma) = \inf_{\sigma'_k\in\frakS}
\cost_k(\tuple{\sigma_{-k},\sigma'_k})$.  A~strategy
profile~$\sigma=(\sigma_i)_{i\in\set n}$ is a \emph{Nash equilibrium}
if for each $k\in\set n$, the~strategy~$\sigma_k$ is a best response
to $(\sigma_i)_{i\in\set n\setminus\{k\}}$.  In~such a case, no~player
has profitable unilateral deviations, \emph{i.e.}, no player alone
can decrease their cost by switching to a different strategy.

Nash equilibria can be defined for subclasses of strategy
profiles. In~particular, a~\emph{blind Nash equilibrium} is a blind
strategy profile~$\sigma$ that is a Nash equilibrium \emph{for
  blind-strategy deviations}: for all $k\in\set n$,
\( \cost_k(\sigma) = \inf_{\sigma'_k\in\frakB}
\cost_k(\tuple{\sigma_{-k},\sigma'_k}) \).  \emph{A~priori}, a~blind
Nash equilibrium need not be a Nash equilibrium for general
strategies.

In~an~NE, once a player deviated from their original strategy in the
strategy profile, the~other players can punish the deviating player,
even if this results in increasing their own costs. Indeed,
the~condition for being an NE only requires that the deviation should
not be profitable to the deviating player. Subgame-Perfect
Equilibria~(SPE) refine~NEs and rule out such non-credible threats by
requiring that, for any path~$h$ in the configuration graph, the
residual strategy profile after~$h$ is an NE.

\begin{example}\label{ex-ex4}
  Consider again the dynamic NCG~$\tuple{\calA,2}$, with the
  arena~$\calA$ of Fig.~\ref{fig-ex}.
  %% a slightly modified version~$\calA'$ of the arena~$\calA$,
  %% where the cost of edge $v_1\to v_3$ is now $x\mapsto 3x$.
  Assume that Player~$1$ plays the blind strategy corresponding
  to~$\pi_3\colon \src \to v_2 \to v_3 \to \tgt$, while Player~$2$
  plays the non-blind strategy~$\sigma$ of Example~\ref{ex-ex2}.
  The~cost for Player~$1$ then is~$10$, while that of Player~$2$
  is~$12$.

  This strategy profile is an NE: Player~$2$ could be tempted to
  play~$\pi_1$, but they would then synchronize with Player~$1$ on edge
  $v_3\to\tgt$, and get cost~$12$ again.
  Similarly, Player~$1$ could be tempted to play~$\pi_1$ instead
  of~$\pi_3$, but in that case strategy~$\sigma$ would tell Player~$2$
  to follow the same path, and the cost for Player~$1$ (and~$2$) would
  be~$16$. Notice in particular that this is not an~SPE, but that the
  blind strategy profile $\tuple{\pi_1, \pi_2}$ (extended to the whole
  configuration tree in the only possible~way) is~an~SPE
  in~$\tuple{\calA,2}$.
\end{example}

%\NM{commented out def. of PoA and PoS}
%% \NM{Do we need to define PoA and PoS? We don't really use them...}
%% We define the following classical notions. 
%% Let~$\calN(\tuple{\calA,n})$ be the set of Nash equilibria
%% in a given dynamic NCG.
%% The \emph{price of anarchy} is 
%% %$\poa(\tuple{\calA,n}) = \frac{\sup_{\pi \in \calN(\tuple{\calA,n})} \soccost(\pi)}{\inf_{\pi \in \frakS^n}\cost(\pi)}$.
%% $\poa(\tuple{\calA,n}) = {\sup_{\sigma \in \calN(\tuple{\calA,n})} \soccost(\sigma)}/{\inf_{\sigma \in \frakS^n}\soccost(\sigma)}$.
%% The \emph{price of stability} is $\pos(\tuple{\calA,n}) = {\inf_{\sigma \in \calN(\tuple{\calA,n})} \soccost(\sigma)}/{\inf_{\sigma \in \frakS^n}\soccost(\sigma)}$.

\medskip

%The rest of this paper is devoted to proving \NM{Wasn't the exact same
%  problem already solved in the literature?}  that whether the social
%optimum is less than or equal to a given bound~$b$ can be decided
%in \PSPACE, and that it~is \NP-hard. \NM{In the asymmetric case, where
%  the number of players is not exponential in the size of the input,
%  do we have an NP-algorithm?}

In~Sections~\ref{section:nash} and~\ref{section:spe}, we~focus on NEs
and SPEs, developing \EXPSPACE and 2\EXPSPACE-algorithms for deciding the existence
of NEs and SPEs respectively of social cost less than or equal to a given
bound. Actually, our approach extends to the $\vec\gamma$-weighted
social cost, where $\vec\gamma\in\bbZ^{\set n}$ are coefficients
applied to the costs of the respective players when computing the
social~cost. As~a consequence, we~can compute best and worst
NEs and SPEs, hence also the price of anarchy and price of
stability~\cite{KP-csr09}.
%\NB{keep PoA and PoS here?}
Before~that, in~Section~\ref{section:socopt},
we~extend classical techniques using blind strategies to compute
the social optimum and prove that NEs always exist.

\section{Socially-optimal strategy profiles}
\label{section:socopt}
To compute a socially-optimal strategy profile, it suffices to find a
path in the concurrent game structure of the given network congestion
game with minimal total cost since one can define a strategy profile
that induces any given path.  Rather than finding such a path in the
concurrent game structure, and in view of Lemma~\ref{lemma-linkMP},
one~can look for one in the abstract weighted graph, thereby reducing
in complexity.  The socially-optimal cost in a dynamic
NCG~$\tuple{\calA,n}$ is thus the cost of a shortest path in the
associated weighted abstract graph~$\calP$ from~$\bar c_\src$
to~$\bar c_\tgt$.
Since $\calP$ has exponential size, we derive complexity upper bounds
for computing a socially-optimal strategy and deciding the associated
decision problem. Moreover, adapting~\cite[Theorem 4.1]{MS-ntwrks12}
which proves \NP-hardness in classical~NCGs, we provide a reduction
from the Partition problem to establish an \NP lower-bound.
\begin{restatable}{theorem}{thmsocopt}
  \label{thm:socopt}
    A socially-optimal strategy profile can be computed in exponential time.
    The~constrained social-optimum problem is in \PSPACE and \NP-hard.
%    When the number of players is fixed, both prolems can be solved in polynomial time.
\end{restatable}

Note, that while~$\calP$ has size $(n+1)^{|V|}$, it is sufficient to
consider paths with a smaller number of transitions when looking for a
shortest path:
\begin{restatable}{lemma}{lemmashortshortest}
  \label{lemma-short-shortest}
  There is a  shortest path (w.r.t.~cost) in~$\calP$
  with size (in~terms of its number of transitions) at most $n\cdot |V|$.
\end{restatable}

\begin{remark}
A~consequence of Lemma~\ref{lemma-short-shortest} is that deciding the
constrained social-optimum problem is in \NP for asymmetric games,
since in that setting the lists of sources and targets of each player
is part of the input, so~that $n$ is polynomial in the size of the input.
However, our~\NP-hardness proof only works in the symmetric case.
\end{remark}

\section{Nash equilibria}
\label{section:nash}
In this section, we study the existence of Nash equilibria and give algorithms to compute
them under given constraints.

\subsection{Existence and computation of (blind) Nash equilibria}
%\SS{Shouldn't it be just NE?}	
\label{section:blindNE}
To prove that blind Nash equilibria always exist, we establish that
dynamic NCGs with blind strategies are potential
games~\cite{R-IJGT73,MS-geb96} which are known to have Nash
equilibria.

Consider a dynamic NCG $\tuple{\calA,n}$, a blind strategy
profile~$\pi$, and let~$N_\pi$ denote the maximum length of the paths
prescribed by~$\pi$. We define the following potential function, which
is an adaptation of that used in \cite{R-IJGT73}:
\[
    \psi(\pi) = \sum_{j=1}^{N_\pi}\sum_{e\in E}\sum_{i=1}^{\load_e(\pi,j)} \edgecost_e(i),
\]
where $\load_e(\pi,j)$ denotes the number of players that take 
edge~$e$ in the $j$-step under~$\pi$, and $\edgecost_e$ is the cost
function on edge~$e$.

Using the above-defined potential function, one can derive an
algorithm to find a Nash equilibrium, by a classical \emph{best-response}
iteration. Starting with an arbitrary blind strategy profile, at~each
step we replace some player's strategy with their best-response,
and we continue as long as some player's cost can be
decreased. When this procedure terminates, the profile at hand is a
blind Nash equilibrium. In~dynamic~NCGs, best~responses exist and can
be computed in polynomial time. Indeed, one can construct a game in
which all players but Player~$i$ follow their fixed strategies given by
profile~$\pi$, using $N_\pi$ copies of the game in order to
distinguish the steps. After the~$N_\pi$-step, all players
in~$\set{n}\setminus\{i\}$ have reached their targets.  Since it is the
only remaining player, the remaining path for Player~$i$ should not be
longer than~$|V|$. Altogether, we obtain the following complexity
upper-bound:
\begin{restatable}{theorem}{thmbralg}
    \label{thm:br-algorithm}
%    In dynamic NCGs, one can find a Nash equilibrium in pseudo-polynomial time.
    In dynamic NCGs, blind Nash equilibria always exist, and 
    we can compute one in pseudo-polynomial time.
    %in exponential time,
    %and in pseudopolynomial time if the number of players is fixed.
\end{restatable}

\begin{remark}
  As an alternative proof to existence of blind NEs, we could have
  bounded the length of outcomes of blind NEs as follows: 
  all players have a strategy realizing cost at most~$|V|\cdot \kappa$, 
  where $\kappa=\max_{e\in E} \edgecost_e(n)$, since the~shortest path from~$\src$
  to~$\tgt$ has length at most~$|V|$, and the cost for a player at
  each step along edge~$e$ is at most~$\kappa$.
  It~follows that   no~path
  along which the cost for some player is larger than~$|V|\cdot
  \kappa$ can be the outcome of a
  blind~NE. As~a consequence, if a dynamic NCG has a blind NE, then it~has 
  one of length at most $|V|\cdot \kappa\cdot
  |V|^n$ (by~removing zero-cycles). Using this bound, we~can transform
  dynamic NCGs into classical congestion games, in
  which blind NEs always exist~\cite{hmrt-tcs11,R-IJGT73}.
\end{remark}
  
%% Note that if an \emph{a priori} bound is known for the size of blind
%% strategies in Nash equilibria, then our games could be reduced to
%% congestion games~\cite{R-IJGT73}, and the existence of Nash equilibria
%% would immediately follow. Our results (see
%% Theorem~\ref{thm:br-algorithm} in Appendix),
%% %\NM{This thm is just below!}
%% actually imply such a bound so this reduction can be
%% considered \emph{a posteriori}.  \NM*{We have a bound: if a player has
%%   cost larger than $|V|\cdot \kappa$ then they can improve. If
%%   all transitions have positive weight, this is a bound on the
%%   paths/blind strategies.}
%% \end{remark}

%Note that the above result is proven for dynamic NCGs restricted to
%blind strategies.  We~show nonetheless
We~now show that blind Nash equilibria are
in fact Nash equilibria.  This~is proved using the observation that
given a blind strategy profile, the~most profitable deviation for any
player can be assumed to be a blind strategy.
\begin{restatable}{lemma}{lemmablindnene}
    \label{lemma:blind-ne-are-ne}
    In dynamic NCGs, blind Nash equilibria are Nash equilibria. 
\end{restatable}

Computing some (blind) Nash equilibrium may not be satisfactory for
two reasons: one might want to compute the best (or the worst) Nash
equilibrium in terms of the social cost; and as
Lemma~\ref{lemma:blind-suboptimal} claims,
%(the full proof is given in appendix)
blind Nash equilibria are suboptimal, \textit{i.e.},
a lower social cost can be achieved by Nash
equilibria with general strategies.  This justifies the study of more
complex strategy profiles in the next subsection.
%The proof is given in the appendix.

\begin{restatable}{lemma}{lemmablindsuboptimal}
    \label{lemma:blind-suboptimal}
    There exists a dynamic NCG with a Nash equilibrium~$\pi$ such that 
    for all blind Nash equilibria~$\pi'$, we~have
    $\cost(\pi) < \cost(\pi')$.
\end{restatable}

The proof (detailed in
%\SS{Do we keep the supporting example in the following figure only?}
Appendix~\ref{app-nash}) is based on the dynamic NCG depicted
on Fig.~\ref{fig:blind-suboptimal}, for which we prove there is a Nash equilibrium
with total cost~$36$, while any \emph{blind} Nash equilibrium has higher social cost.

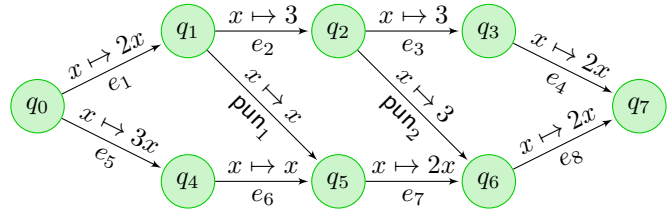
\begin{figure}[htbp]
		\centering
		\begin{tikzpicture}[shorten >=1pt,node distance=1cm,on grid,auto, el/.style = {inner sep=3pt, align=left, sloped}]
		\node[rond,vert] (q0) at (0,0) {$q_0$};
		\node[rond,vert] (q1) at (2,1) {$q_1$};
		\node[rond,vert] (q2) at (4,1) {$q_2$};
		\node[rond,vert] (q3) at (6,1) {$q_3$};
		\node[rond,vert] (q4) at (2,-1) {$q_4$};
		\node[rond,vert] (q5) at (4,-1) {$q_5$};
		\node[rond,vert] (q6) at (6,-1) {$q_6$};
		\node[rond,vert] (q7) at (8,0) {$q_7$};
		\begin{scope}[-latex']
		\draw (q0) edge node[el, above] {$x\mapsto 2x$} node [el, below]{$e_1$} (q1);
		\draw (q0) edge node[el, above] {$x\mapsto 3x$} node [el, below]{$e_5$} (q4);
		\draw (q1) edge node[el, above] {$x\mapsto 3$} node [el, below]{$e_2$} (q2);
		\draw (q2) edge node[el, above] {$x\mapsto 3$} node [el, below]{$e_3$} (q3);
		\draw (q3) edge node[el, above] {$x\mapsto 2x$} node [el, below]{$e_4$} (q7);
		\draw (q4) edge node[el, above] {$x\mapsto x$} node [el, below]{$e_6$} (q5);
		\draw (q5) edge node[el, above] {$x\mapsto 2x$} node [el, below]{$e_7$} (q6);
		\draw (q6) edge node[el, above] {$x\mapsto 2x$} node [el, below]{$e_8$} (q7);
		\draw (q1) edge node[el, above] {$x\mapsto x$} node [el, below]{$\pun_1$} (q5);
		\draw (q2) edge node[el, above] {$x\mapsto 3$} node [el, below]{$\pun_2$} (q6);
		\end{scope}
                \end{tikzpicture}
		\caption{An arena on which blind Nash equilibria are sub-optimal.}\label{fig:blind-suboptimal}
	\end{figure}

\subsection{Computation of general Nash equilibria}
\label{section:nash-optimal}
\subparagraph*{Characterization of outcomes of Nash Equilibria.}
% In this section, we first prove a characterization of the set of outcomes
% of Nash equilibria of a given dynamic NCG.
Let us consider a dynamic NCG~$\tuple{\calA,n}$, and the corresponding
game structure~$\calS = \tuple{C, T, M, U}$.
Given two configurations~$c,c'$ with~$c\Rightarrow c'$, we let
$\cost_i(c,c')$ denote the cost of Player~$i$ on this transition from~$c(i)$ to~$c'(i)$.
We define $\dev_i(c,c')$ as the set of all configurations
reachable when all players but Player~$i$ choose moves prescribed by the
given transition~$c\Rightarrow c'$:
\[
\dev_i(c,c') = \{ c'' \in C \mid c \Rightarrow c''\text{ and } \forall j \in \set{n} \setminus\{i\}.\ c''(j) =c'(j)\}.
\]

The \emph{value} of configuration~$c$ for Player~$i$
is
\(
\val_{i,c}=\sup_{\sigma_{-i} \in \frakS^{n-1}} \inf_{\sigma_i\in
  \frakS} \cost_i((\sigma_{-i}, \sigma_i), c).
\)
Note that the value corresponds to the value of the zero-sum game
where Player~$i$ plays against the opposing coalition, starting at~$c$.
By~\cite{KBBEGRZ-tcs08}, those values can be computed in polynomial
time in the size of the game. Here the game is a 2-player game with
state space $|V|\times \set{n-1}^{|V|}$, keeping track of the position
of Player~$i$ and the abstract position of the coalition. It~follows
that each $\val_{i,c}$ can be computed in exponential time in the size of
the input~$\tuple{\calA,n}$.
%
%~$\calS$, which is doubly-exponential\NM{was ``exponential''} in the size
%of the input~$\tuple{A,n}$.
%
Moreover, memoryless optimal strategies exist (in~$\calS$), that~is,
the~opposing coalition has a memoryless strategy~$\sigma_{-i}$ to ensure a
cost of at least~$\val_{i,c}$ from~$c$.
%\todo{right?\\ NM: Yes}

The characterization of Nash equilibria outcomes is given in the following lemma.
\begin{restatable}{lemma}{lemmanecharact}
    \label{lemma:ne-charact}
    A path~$\rho$ in $\tuple{\calA,n}$ is the outcome of a Nash equilibrium if, and only~if, 
    \[
    \forall i \in \set{n}.\
    \forall 1\leq l < |\rho|.\
    \forall c \in \dev_i(\rho(l),\rho(l+1)).\quad
        \cost_i(\rho_{\geq l}) \leq \val_{i,c} + \cost_i(\rho(l), c).
    \]
\end{restatable}
%\SS{Do we need to keep this intuition anymore?}
%% NM: giving intuition cannot hurt...
The intuition is that
if the suffix
$\cost_i(\rho_{\geq l})$ of~$\rho$ has cost more than $\val_{i,c} + \cost_i(\rho(l), c)$,
then Player~$i$ has a profitable deviation regardless of the strategy of the
opposing coalition,
since $\val_{i,c}$ is the maximum cost
that the coalition can inflict to Player~$i$ at configuration~$c$ where the deviation is observed.
The lemma shows that the absence of such a suffix means that a Nash equilibrium with given outcome
exists, which the proof constructs.

\begin{proof}
	Consider a Nash equilibrium~$\sigma = (\sigma_i)_{i \in \set{n}}$
	with outcome~$\rho$.
	Consider any player~$i$, and any strategy~$\sigma'_i$ for this player. 
	Let~$\rho'$ denote the outcome of~$\sigma[i \rightarrow \sigma_i']$.
	Let~$l$ denote
	the index of the last configuration where~$\rho$ and~$\rho'$ are identical.
	Since~$\sigma$ is a Nash equilibrium, we have
	$\cost_i(\rho) \leq \cost_i(\rho'$), that is,
	\[
	\cost_i(\rho_{\geq l}) \leq \cost_i(\rho(l),\rho'(l+1)) +
	\cost_i(\sigma[i \rightarrow \sigma_i'], \rho'_{\leq l+1})
	\]
	where $\cost_i(\sigma[i \rightarrow \sigma_i'], \rho'_{\leq l+1})$
	is the cost for Player~$i$ of the outcome of the residual strategy
	$(\sigma[i \rightarrow \sigma_i'])^{\rho'_{\leq l+1}}$.
	%\SS{Do we have $\cost_i(\sigma,h)$ notation yet?}
	%\NM*{Why $\cost_i(\rho(l),\rho(l+1))$ and
	%       not $\cost_i(\rho(l),\rho'(l+1)) $?}
	Since the choice of~$\sigma_i'$ is arbitrary here, we have,
	\[
	\cost_i(\rho_{\geq l}) \leq \cost_i(\rho(l),\rho'(l+1)) +
	\inf_{\sigma_i' \in \frakS}\cost_i(\sigma[i \rightarrow \sigma_i'],
	\rho'_{\leq l+1}).
	\]
	Moreover, we have 
	$\inf_{\sigma_i' \in \frakS}
	\cost_i(\pi[i \rightarrow \sigma_i'], \rho'_{\leq l+1})
	= \inf_{\sigma_i' \in \frakS}
	\cost_i(\pi[i \rightarrow \sigma_i'], \rho'(l+1))$
	since memoryless strategies suffice 
	to minimize the cost~\cite{KBBEGRZ-tcs08}.
	We then have
	\[
	\inf_{\sigma_i' \in \frakS} \cost_i(\pi[i \rightarrow \sigma_i'],
	\rho'(l+1)) \leq \sup_{\sigma_{-i} \in \frakS^{n-1}}
	\inf_{\sigma_i \in \frakS} \cost_i((\sigma_{-i},\sigma_i),
	\rho'(l+1)).
	\]
	We obtain the required inequality
	\begin{xalignat*}1
		\cost_i(\rho_{\geq l}) &\leq \cost_i(\rho(l),\rho'(l+1)) +
		\sup_{\sigma_{-i} \in \frakS^{n-1}}
		\inf_{\sigma_i \in \frakS}
		\cost_i((\sigma_{-i},\sigma_i), \rho'(l+1))\\
		&\leq    \cost_i(\rho(l),c) + \val_{i,c}.
	\end{xalignat*}
	
	Conversely, consider a path~$\rho$ that satisfies the condition. We are going to construct
	a Nash equilibrium having outcome~$\rho$. The idea is that players will follow~$\rho$,
	and if some player~$i$ deviates, then the coalition~$-i$ will apply a joint strategy to maximize the
	cost of Player~$i$, thus achieving at least~$\val_{i,c}$, where~$c$ is the first configuration where deviation
	is detected.
	
	Let us define the punishment function $\calP_\rho
	\colon\Paths(\tuple{\calA,n})\rightarrow \set{n}\cup \{\bot\}$
	which keeps track of the deviating players and the step where such
	a player has deviated.  For path~$h'=h(c,w,c')$, we write
	\[
	\calP_\rho(h') =
	\left\{\!\!
	\begin{array}{l@{}l}
	\bot & \text{ if } h' \prefix \rho,\\
	i & \text{ if } h \prefix \rho, h(c,w,c') \not \prefix \rho, \text{ and }i \in \set{n} \text{ min. s.t. } c'(i) \neq \rho(|h|+1)(i),\\
	\calP_\rho(h) & \text{ otherwise}.
	\end{array}
	\right.
	\]
	%    \NM*{what if several players deviated?}
	Intuitively, $\bot$~means that no players have deviated from~$\rho$ in the current path.
	If~$\calP_\pi(h)=j$, then Player~$j$ was among the first players to deviate from~$\rho$ in the path~$h$;
	so for some~$l$, $h(l)(j) = \rho(l)(j)$ but $h(l+1)(j)\neq \rho(l+1)(j)$.
	Notice that if several players deviate at the same step, there are no conditions to be checked, and the strategy
	can be chosen arbitrarily.
	For~each configuration~$c$ and coalition~$-i$, let~$\sigma_{-i,c}$ be the strategy of coalition~$-i$ maximizing
	the cost of Player~$i$ from configuration~$c$; thus achieving at least~$\val_{i,c}$.
	Player~$j$'s strategy in this coalition, for~$j \neq i$, is denoted~$\sigma_{-i,c,j}$.
	For path~$h'=h(c,w,c')$, define
	\[
	\tau_i(h') =
	\left\{
	\begin{array}{ll}
	(c'(i),m(i),c''(i)) & \text{if } \calP_\rho(h') = \bot,
	\rho(|h'|+1) = (c',w',c''), \\& \text{ and }
	m \in E^n \text{ is such that } T(c',m) = (w',c''),\\
	\text{arbitrary} & \text{if } \calP_\rho(h')= i,\\
	\sigma_{-j,c,i}(h') & \text{if } \calP_\rho(h')= j \text{ for some } j\neq i.
	\end{array}
	\right.
	\]
	The first case ensures that the outcome of the
	profile~$(\tau_i)_{i \in \set{n}}$ is~$\rho$.  The third case
	means that Player~$i$ follows the coalition
	strategy~$\sigma_{-j,c}$ after Player~$j$ has deviated to
	configuration~$c$. The second case corresponds to the case where
	Player~$i$ has deviated: the precise definition of this part of
	the strategy is irrelevant.
	
	Let us show that this profile is indeed a Nash equilibrium. Consider any player $j \in \set{n}$ and
	any strategy~$\tau_j'$. Let~$\rho'$ denote the outcome of~$(\tau_{-j},\tau_j')$, and
	$l$ the index of the last configuration where $\rho$ and~$\rho'$ are identical.
	We have
	\begin{xalignat*}1
		\cost_j((\tau_{-j},\tau_j')) &= \cost_j(\rho_{\leq l}) + \cost_j(\rho(l),\rho'(l+1)) + \cost_j((\tau_{-j},\tau_j), \rho'_{\leq l+1})\\
		&\geq  \cost_j(\rho_{\leq l}) + \cost_j(\rho(l),\rho'(l+1)) + \val_{j,\rho'(l+1)(j)}\\
		&\geq \cost_j((\tau_i)_{i \in \set{n}}),
	\end{xalignat*}
	where the second line follows from the fact that the coalition switches to a strategies ensuring a cost of at least
	$\val_{j, \rho'(l)(j)}$ at step~$l$; and the third line is obtained by assumption.
	This shows that~$(\tau_i)_{i \in \set{n}}$ is indeed a Nash equilibrium and concludes the proof.
\end{proof}

\subparagraph*{Algorithm.}
We define a graph that describes the set of outcomes of Nash equilibria by augmenting the 
$n$-weighted configuration graph $\calM=\tuple{C, T}$. 
For~any real vector~${\vec\gamma=(\gamma_i)_{i\in\set n}}$,
%\SS{Need to change this from $\vec\lambda$ to (say) $\vec\gamma$}
%(whose role will be made clear later),
we~define the weighted graph~$\negraph = \tuple{C', T'}$
with~$C' = C \times (\set{Y}\cup\{0,\infty\})^n$
%\SS{Technically, $0$ should be included, no?},
where ${Y=|V|\cdot \kappa}$,
 and $T' \subseteq C' \times \mathbb{N} \times C'$;
 %Note
   remember that all players have a strategy realizing
   cost at most~$Y$ in~$\tuple{\calA,n}$.
%:
%the~shortest path from~$\src$ to~$\tgt$ has length at most~$|V|$,
%and the cost for a player at each step along edge~$e$ 
%is at most~$\kappa$.
%That is, the state space is the set of pairs
%of configurations and a value for each player, and~$T''$ is the set of transitions defined below.
The~initial state is~$(c_\src,\infty^n)$.
The~set of transitions~$T'$ is defined as follows:
$((c,b),z,(c',b')) \in T'$ if, and only~if,
there exists $(c,w,c') \in T$, $z = \vec{\gamma}\cdot w$
(where~$\cdot$ is dot product),
and for all $i \in \set{n}$,
\begin{equation}
    \label{eqn:bi}
    b_i'= \min(b_i - w_i, \min_{c'' \in \dev_i(c(i),c'(i))} \cost_i(c,c'') + \val_{i,c''} - w_i).
\end{equation}
Notice that by definition of~$C'$, $b'_i$ must be nonnegative for all~$i\in\set n$,
so there are no transitions
$((c,b),z,(c',b'))$ if the above expression is negative for some~$i$.
Notice also that the size of $\negraph$ is doubly-exponential in that of the
input $\tuple{\calA,n}$, since this is already the case for~$C$,
%; we have $|C|=|V|^n$ and $n$ is given in binary,
while~$Y$ is singly-exponential.
%However, the~size of $\negraph$ is
%singly exponential if~$n$ is given in unary.

Intuitively, for any path~$\rho$ that visits
%\SS{ends or passes through?}
some state~$(c,b)$ in this graph,
in order for~$\rho$
to be compatible with a Nash equilibrium, each player~$i$ must have cost
no more than~$b_i$ in the rest of the path.
In fact, the second term of the minimum in~\eqref{eqn:bi}
is the least cost Player~$i$ could guarantee by not following $(c,w,c')$ but going to
some other configuration~$c'' \in \dev_i(c,c')$, so the bound~$b_i$ is used to guarantee
that these deviations are not profitable.
The definition of~$b_i'$ in~\eqref{eqn:bi} is the minimum of~$b_i-w_i$ 
and the aforementioned quantity
since we check both the previous bound~$b_i$, updated
with the current cost~$w_i$ (which gives the left term), and the non-profitability of a deviation
at the previous state (which is the right term).
If this minimum becomes negative, this precisely means that at an earlier point in 
the current path, there was a strategy for Player~$i$ which was more profitable than 
the current path regardless of the strategies of other players; so the current path cannot
be the outcome of a Nash equilibrium. This is why the definition of $\negraph$ restricts
the state space to nonnegative values for the $b_i$.

We prove that
%solving the
%\SS{Why so we use ``the'' optimal? Shouldn't it be $\gamma$-optimal?}
%best- and worst Nash equilibrium problems reduces
computing the cost of a
Nash equilibrium minimizing the $\vec\gamma$-weighted social cost
reduces
to computing a shortest path in~$\negraph$.
%Actually, we~compute the
%cost of a
%$\vec\gamma$-minimal
%Nash equilibrium minimizing the $\vec\gamma$-weighted social cost.
In~particular, letting~$\gamma_i=1$ for all $i\in\set n$,
a~$\vec{\gamma}$-minimal Nash equilibrium is a best Nash equilibrium
(minimizing the social cost), while taking $\gamma_i=-1$ for all $i\in\set
  n$, we get a worst Nash equilibrium (maximizing the social cost).%
\begin{restatable}{theorem}{thmminne}\label{thm-minNE}
    For any dynamic NCG~$\tuple{\calA,n}$ and vector~$\vec{\gamma}$,
    the cost of the shortest path from~$(c_\src,\infty^n)$ 
    to some $(c_\tgt,b)$ in~$\negraph$
    is the cost of a $\vec{\gamma}$-minimal
    %\SS{Has it been defined?}
    Nash equilibrium.
\end{restatable}

\begin{proof}
	We show that for each path of~$\tuple{\calA,n}$ from~$c_\src$ to~$c_\tgt$, 
	there is a path in~$\negraph$ from $(c_\src,\infty^n)$ to some~$(c_\tgt,b)$
	with the same cost, and vice versa.
	
	Consider a Nash equilibrium~$\pi=(\sigma_j)_{j \in \set{n}}$ with 
	outcome~$\rho=(c_j,w_j,c_{j+1})_{1\leq j < l}$.
	We~build a sequence~$b_1,b_2,\ldots$ such that
	$\rho' = ((c_j,b_j),\vec{\gamma}\cdot w_j, (c_{j+1},b_{j+1}))_{1\leq j < l}$ is a path
	of~$\negraph$.
	We~set~$b_1(j) = \infty$ for all~$j \in \set{n}$.
	For~$j\geq 1$, define
	\[
	b_{j+1}(i) = \min\left(b_{j}(i) - w_j(i), 
	\min_{c'' \in \dev_j(c_j(i),c_{j+1}(i))} \cost_i(c_j,c'') + \val_{i,c''} - w_j(i)\right). %\cost_i(c_j,c_{j+1})).
	\]
	We are going to show that for all~$1\leq j \leq l$, 
	$\cost_i(\rho_{\geq j}) \leq b_j$, which shows that~$b_j\geq 0$, and thus~$\rho'$
	is a path of~$\negraph$. 
	
	We show this by induction on~$j$. This is clear for~$j=1$. 
	Assume this holds up to~$j\geq 1$.
	We have, by induction that~$\cost_i(\rho_{\geq j}) \leq b_j(i)$ for all $i \in \set{n}$.
	Moreover, since~$\pi$ is a Nash equilibrium, by Lemma~\ref{lemma:ne-charact},
	%    \SS*{Don't you want to put limits under the min?}
	\[
	\forall i \in \set{n},
	\cost_i(\rho_{\geq j}) \leq
	\min_{c'' \in \dev_i(\rho(j),\rho(j+1))}
	\val_{i,c''} + \cost_i(\rho(j), c'').
	\]
	Therefore,
	\begin{xalignat*}1
		\cost_i(\rho_{\geq j+1}) &= \cost_i(\rho_{\geq j}) - w_j(i)\\
		& \leq 
		\min(b_j(i) - w_j(i), \min_{c'' \in \dev_i(\rho(j),\rho(j+1))}
		\val_{i,c''} + \cost_i(\rho(j), c'') - w_j(i))
	\end{xalignat*}
	as required,
	and both paths have the same $\vec{\gamma}$-weighted cost.
	
	Consider now a path~$((c_i,b_i),z_i,(c_{i+1},b_{i+1}))_{1\leq i <l}$ in
	$\negraph$. By the definition of~$\negraph$, there exists~$w_1,w_2,\ldots$ such
	that~$\rho =(c_j,w_j,c_{j+1})_{1\leq j<l}$ is a path of~$
	\tuple{\calA,n}$, and~$z_j = \vec{\gamma}\cdot w_j$.
	So it only remains to show that that $\rho$ is the outcome of a  Nash equilibrium.
	We will show that $\rho$ satisfies the criterion of Lemma~\ref{lemma:ne-charact}.
	%which means that there exists a Nash equilibrium with outcome~$\rho$.
	We show by backwards induction on~$1\leq j \leq l$
	%\SS{Changed the induction indexes, it was $1 \leq j < l$}
	that for all $i \in \set{n}$,
	\begin{enumerate}
		\item $\cost_i(\rho_{\geq j}) \leq b_j(i)$,
		\item $\cost_i(\rho_{\geq j}) \leq \min_{c'' \in \dev_i(\rho(j),c'')} 
		\cost_i(\rho(j),c'') + \val_{i,c''}$.
	\end{enumerate}
	For~$j=l$, we have~$\cost_i(\rho_{\geq l}) = 0$ so this is trivial.
	Assume the property holds down to~$j+1$
	%\SS{check this line, I guess it's Okay}
	for some~$1\leq j<l$. By~induction hypothesis, we have 
	\[
	\cost_i(\rho_{\geq j+1}) \leq b_{j+1}(i) = \min\left(b_j(i) - w_j(i),
	\min_{c'' \in \dev_i(\rho(j),c'')} \cost_i(\rho(j),c'') + \val_{i,c''}
	- w_j(i)\right).
	\]
	Therefore,
	\[
	\cost_i(\rho_{\geq j}) = \cost_i(\rho_{\geq j+1}) + w_j(i) \leq 
	\min\left(b_j(i),
	\min_{c'' \in \dev_i(\rho(j),c'')} \cost_i(\rho(j),c'') + \val_{i,c''}\right),
	\]
	as required.
	By Lemma~\ref{lemma:ne-charact}, $\rho$ is the outcome of a Nash equilibrium.
\end{proof}

%\SS{Not needed the following paragraph anymore, right?}
%This way,
%\NM{updated the following paragraph; please check}
Thanks to Theorem~\ref{thm-minNE}, 
we~can compute the costs of the best and worst~NEs
of~$\tuple{\calA,n}$ in exponential space.
We~can also decide the
existence of an NE with constraints on the costs (both social and
individual), by non-deterministically guessing an outcome and checking
in~$\negraph$ that it is indeed an~NE.
We obtain the following conclusion:
%. Note that the size of~$\negraph$ is only pseudopolynomial
%when~$n$ is fixed\NM{I suggest we omit the case ``$n$ fixed'' (here and in corollary)};
%so the main parameter is the number of players.
%and a pseudopolynomial complexity is obtained by fixing it.
%
\begin{restatable}{corollary}{coroptne}
    \label{corollary:optimal-ne}
    In dynamic NCGs, the constrained Nash-equilibrium problem is in \EXPSPACE.
    %,
    %and can be solved in pseudo-polynomial time when the number of players is fixed.
    %Both problems can be solved in pseudopolynomial time when the number of players is fixed.
\end{restatable}

\begin{proof}
	As noted earlier, the number of vertices in $\negraph$ is doubly exponential
	since
        %$Y$ is exponential
        %\SS{Couldn't we avoid writing ``$Y$ is exponential" here, as (1) it has been mentioned earlier, (2) $|C|$ is more affecting factor}
        %and
        $|C| = |V|^n$ is doubly exponential.
        Storing a configuration and computing its successors can be
        performed in 
	%A configuration can be represented, and successors can be computed in
        exponential space.
	One~can thus guess a path of size at most the size of the graph and check whether its cost is less than the given bound.
	This can be done using exponential-space counters, and provides us with an \EXPSPACE algorithm.
	%    This complexity becomes psuedo-polynomial when~$n$ is fixed.\SS{Do we talk about the case where $n$ is fixed anymore?}
\end{proof}

%Notice that another approach to solving the constrained-NE problem
%would consist in non-deterministically guessing a path
%in~$\negraph$ and checking that it
%satisfies the constraints.
%
Note that one can effectively compute
%a $\vec{\gamma}$-minimal
a Nash-equilibrium strategy profile satisfying the constraints
in doubly-exponential time by finding the shortest path of~$\negraph$,
and applying the construction of (the~proof~of) Lemma~\ref{lemma:ne-charact}.
%\NM{Do we mean ``the proof of Lemma~\ref{lemma:ne-charact}''? Is it understandable without reading the proof}

\begin{remark}
%\NM*{Complexity of asymmetric case (equiv. (?) number of players given in unary)}
  The~exponential complexity is due to the encoding of the number of
  players in binary. If we consider asymmetric NCGs, in which the
  source-target pairs would be given explicitly for all players, the
  size of $\negraph$ would be singly-exponential, and
  the constrained Nash-equilibrium problem 
  would be in \PSPACE.
\end{remark}

%% Combining the previous algorithms, one can compute the costs of best and worst Nash equilibria,
%% and the social optimum. This yields the following results.
%% \begin{corollary}
%%     \label{corollary:poa}
%%     Price of anarchy and price of stability can be computed in doubly-exponential time
%%     in dynamic NCGs, and in pseudopolynomial time when the number of players is fixed.
%% \end{corollary}

%\input{nash}

\section{Subgame-perfect equilibria}
\label{section:spe}

In this section, we
characterize the outcomes of SPEs and decide the existence of SPEs
with constraints on the social cost.
We~follow the approach of~\cite{BBGRB-concur19},
extending~it to the setting of  concurrent weighted games, which we
need to handle dynamic NCGs.

\subparagraph{Characterization of outcomes of~SPE.}
%We show a characterization of the set of outcomes of Subgame Perfect Equilibria of a given dynamic NCG.

Consider a dynamic NCG~$\tuple{\calA,n}$, and the associated
configuration graph~$\calM=\tuple{C,T}$. We~partition the set~$C$ of
configurations into~$(X_j)_{0\leq j\leq n}$ such that a
configuration~$c$ is in~$X_j$ if, and only~if, $j=\#\{i\in\set n\mid
c(i)=\tgt\}$. Since~$\tgt$ is a sink state in~$\calA$, if~there is a
transition from some configuration in~$X_j$ to some configuration
in~$X_k$, then $k\geq j$.
We~define $X_{\geq j}=\bigcup_{i\geq j} X_i$,
%We~(abusively) say that an edge~$e=(c,w,c')$ belongs to~$X_j$ if~$c$~does.
$\SE_{j}=\{(c,w,c')\in T \mid c\in X_{j}\}$
and $\SE_{\geq j}=\{(c,w,c')\in T \mid c\in X_{\geq j}\}$.
%\SS{Need for $\SE_j$ might arise}\NM{Indeed. It is defined somewhere later}
%\SS{Need to say about the topological order between $X_j$'s because of the one-directional edges}

Following~\cite{BBGRB-concur19}, we inductively define a sequence
$(\lambda^{j^*})_{0 \leq j \leq n}$, where each $\lambda^{j^*} =
\tuple{\lambda^{j^*}_i}_{i \in \set{n}}$ is a $n$-tuple of labeling
functions
$\lambda^{j^*}_i\colon \SE_{\geq j} \to \bbN \cup \{-\infty, +\infty\}$.
%such that a $\lambda^{j^*}_i$ maps an edge $\mathbf e$ from
%$X_{\geq j}$ to $\bbN \cup \{0, \bot, \infty\}$.  We write this as
%$\lambda^{j^*}_i : X_{\geq j} \rightarrow \{0, \bot, \infty\}$ (with
%the abusive notation).
%
This sequence will be used to characterize outcomes of SPEs through
the notion of \emph{$\lambda$-consistency}:
\begin{definition}\label{def:lambdaconsistent}
%% NM: this is a crucial definition, but we may remove the "def" envir.
%% if needed
  Let~$j\leq n$, and $\lambda=(\lambda_i)_{i\in\set n}$ be a family of
  functions such that $\lambda\colon \SE_{\geq j}\to \bbN
  \cup\{-\infty,+\infty\}$
  %\SS{Are we assuming $0 \in \bbN$? :P}\NM{Yes}.
  Let~$c\in X_{\geq j}$.
  A~finite path~$\rho=(t_k)_{1\leq k<|\rho|}$ from~$c$ ending in~$c_{\tgt}$ is said to be
  $\lambda$-consistent whenever for any~$i\in\set n$ and any~$1\leq
  k<|\rho|$, it~holds $\cost_i(\rho_{\geq k}) \leq
  \lambda_i(t_k)$. We~write $\LambdaSet_\lambda(c)$ for the set of all
  $\lambda$-consistent paths from~$c$.
\end{definition}

%% In the discussion, we sometimes omit unnecessary indexes when we talk about a general edge labeling map $\lambda$ (not specific to the $\lambda$'s we consider in the algorithm), for example for $\lambda = \tuple{\lambda_i}_{i \in \set{n}}$ with each $\lambda_i$ defined on some $X_{\geq j}$, we say a path $\rho = (t_l)_{1 \leq l < |\rho|}$ starting from some configurations $c \in X_{\geq j}$ to be $\lambda$-consistent if
%% \begin{align*}
%% \forall i \in \set{n} ~\forall 1 \leq l < |\rho|~\cost_i(\rho_{\geq l}) \leq \lambda_i(t_l)
%% \end{align*}

%% We denote the set of $\lambda$-consistent paths from $c$ as~$\LambdaSet(c)$.
%% \NM{$\lambda$ should appear in the notation: $\LambdaSet_\lambda(c)$?}

We~now define~$\lambda^{j^*}$ for all~$0\leq j\leq n$ in such a way
that, for all~$c\in X_{\geq j}$, $\LambdaSet_{\lambda^{j^*}}(c)$ is the
set of all outcomes of SPEs in the subgame rooted at~$c$.
%\NM{should we  define ``subgame''?}
The~case where $j=n$ is simple: we~have
$X_{\geq n} = \{c_\tgt\}$ and $\SE_{\geq n} = \{(c_\tgt,0^n,c_\tgt)\}$;
there~is a single path, which obviously is the outcome of an~SPE since
no deviations are possible. For~all~$i\in\set n$, we~let
$\lambda^{n^*}_i(c_\tgt,0^n,c_\tgt) =0$.

%% Now back to the algorithm, we define our labeling functions $\lambda^{j^*}_i$'s as follows:

%% For $j = n$, by definition, $X_{\geq j} = \{c_{\tgt}\}$, and we take $\lambda^{n^*}_i(\mathbf e)  = 0$ for $\mathbf e \in X_{\geq j}$.

Now, fix~$j<n$, assuming that $\lambda^{(j+1)^*}$ has been defined.
In~order to define~$\lambda^{j^*}$, we~introduce an intermediary sequence
$(\mu^k_i)_{k\geq 0,i\in\set n}$, 
with $\mu^k_i\colon \SE_{\geq j} \to \bbN\cup\{-\infty, +\infty\}$,
of which $(\lambda^{j^*}_i)_{i\in\set n}$ will be the limit. 

Functions~$\mu^k_i$ mainly operate on~$\SE_j=\SE_{\geq j}\setminus \SE_{\geq j+1}$: 
for~any~$\mathbf e\in \SE_{\geq j+1}$, we~let $\mu^k_i(\mathbf e)=\lambda^{(j+1)^*}_i(\mathbf e)$.
Now, for~$\mathbf e=(c,w,c')\in \SE_j$, $\mu^k_i(e)$ is defined inductively as follows:
\begin{itemize}
\item $\mu_i^0(\mathbf e)=0$ if $c(i)=\tgt$, and $\mu_i^0(\mathbf e)=+\infty$ otherwise;
\item for~$k>0$, %assuming~$\mu^{k-1}$ has been defined,
  $\mu^k$ is defined from~$\mu^{k-1}$ following three cases:
   if~$c(i)=\tgt$, then $\mu_i^k(\mathbf e)=0$;
   if $\LambdaSet_{\mu^{k-1}}(c')=\emptyset$ for some~$(c,w',c')\in T$, then $\mu_i^k(\mathbf e)=-\infty$;
   otherwise,
   %\NM{use $\dev_i$ notation under $\min$?}
  \[
  \mu_i^k(\mathbf e)=
%  \begin{cases}
%  \min\limits_{\substack{(c,w',c'')\in T\\ \forall j\not=i.\ c''(j)=c'(j)}}
  \mathop{\min\vphantom{\sup}}\limits_{c''\in\dev_i(c,c')} \ 
  \sup_{\rho\in\LambdaSet_{\mu^{l-1}}(c'')} ( \cost_i(c,c'') + \cost_i(\rho) )
%& \text{if }\forall (c,\tilde{w}, \tilde{c}) \in T,~\LambdaSet_{\mu^{l-1}}(\tilde{c}) \neq \emptyset\\
%  \bot &\text{ otherwise}
%  \end{cases}
  \]  
\end{itemize}
We~can then prove that
for any~$e\in \SE_{\geq j}$ and any~$k>0$, $\mu^k_i(e)\geq \mu^{k-1}_i(e)$.
%\end{lemma}
It~follows that the sequence~$(\mu^k)_{k\geq 0}$ stabilizes, and we can define
$\lambda^{j^*}$ as its limit. 
Let $\LambdaSet^*=\LambdaSet_{\lambda^{0^*}}$. Then:
\begin{restatable}{theorem}{lambdaconsistent}
	\label{thm:lambdaconsistency}
	A path $\rho$ in $\calG = \tuple{\calA, n}$ is the outcome of
        an SPE if, and only~if, $\rho \in \LambdaSet^*(c_{\src})$.
        %$\LambdaSet^*(c_\src)$.
%        \NM{$\LambdaSet^*$ is $\LambdaSet^{0^*}$? Perhaps we can introduce this notation}
\end{restatable}

\subparagraph*{Algorithm.}

It~remains to compute the sequence~$(\mu^k)_{k\geq 0}$ (which will
include checking non-emptiness of the corresponding $\LambdaSet$-sets),
and to bound the stabilization time.
To this aim, with any
family~$\mu=(\mu_i)_{i\in\set n}$ of functions as above
and any configuration~$c$, we~associate an infinite-state \emph{counter graph}
$\bbC[\mu,c]=\tuple{C',T'}$
%\NM{removed initial state}
to capture all $\mu$-consistent paths from~$c$:
\begin{itemize}
\item the~set of vertices
  is~$C'= C\times (\bbN\cup\{+\infty\})^{\set n}$;
  %\{(d,b)\in C\times (\bbN\cup\{+\infty\})^{\set n} \mid c
  %\Rightarrow^* d% \text{ and } b\in [0, Y]\cup\{+\infty\}
  %\}$
%  where~$Y=6|V|^n\cdot (\max_{e\in E} \edgecost_e(n))^{|V|}$;
%  \NM{check  bound on $Y$?}
%\NM{letter $Y$ already used for NEs. Change?}
\item $T'$~contains all edges
  $((d,b),w,(d',b'))$ for which $(d,w,d')$ is an edge of~$\calM$ and for
  all~$i\in\set n$, $b'(i)=0$ if~$d(i)=\tgt$, and $b'_i=\min\{b_i-w_i,
  \mu_i(d,w,d')-w_i\}$ otherwise (provided that~$b'_i\geq 0$ for
  all~$i$, in~order for~$(d',b')$ to be an edge of~$\bbC[\mu,c]$).
\end{itemize}
%\NM*{definine infinite-state counter graph and then bound $b$}
With the initial configuration~$c$, we~associate $b^c$ such that
$b_i^c=0$ if~$c(i)=\tgt$ and $b_i^c=+\infty$ otherwise: this
configuration imposes no constraint, since no edges has been
taken~yet.
Intuitively, in configuration~$(d,b)$, $b$~is used to enforce
$\mu$-consistency: each edge taken along a path imposes a constraint
on the cost of the players for the rest of the path; this constraint
is added to the constraints of the earlier edges, and propagated along
the path. We~can prove that the number of reachable states
from~$(c,b^c)$  in~$\bbC[\mu,c]$,
which we denote with~$|C'|_r$,
is bounded
by $|C|\cdot (n\cdot |V|^n\cdot \kappa)^{|V|}$.
%by~$|C|\cdot 6n|V|^{n}\kappa^{|V|}$
%\SS{Need to update with new value (after Nico's check maybe?)}.

Computing~$\lambda^{j^*}$ from~$\lambda^{(j+1)^*}$ amounts to
inductively computing~$(\mu^{k+1}_i)_{i\in\set n}$ from~$\mu^k$ for
edges~$e=(c,w,c')\in \SE_j$, until stabilization.
Since~$\bbC[\mu^{k},d]$ can be proved to capture $\mu^{k}$-consistent
paths from~$d$, the~computation mainly amounts to checking the
existence of paths in such counter graphs, which can be performed in
doubly-exponential space.  Stabilization can be shown to occur within
$|V|(1+n\cdot\kappa\cdot |E|^n)$ steps.
In~the end:
\begin{restatable}{theorem}{existsSPE}
	\label{thm: existsSPE}
	The existence of SPEs in a dynamic NCG can be decided in 2\EXPSPACE.
        %\NM{or 2\EXPSPACE?} 
\end{restatable}

\begin{remark}
  Again, our algorithm is not specific to the symmetric setting of our
  dynamic NCGs; in an asymmetric context, where the number of players
  would be given in unary, our algorithm would run in~\EXPSPACE.
  %\fbox{asymmetric case}
\end{remark}

\subparagraph{Existence of constrained SPEs.}
%Once we get the answer ``yes" to the $\exists$-SPE problem, we can solve constraint SPE problem using the counter graph $\bbC[\lambda^*, c_\src]$, where the constraint of the problem is on one or multiple player's cost.
The algorithm above can be extended to compute the cost of the best
and worst SPEs, and to include constraints on the costs (both social
and individual) of the SPEs we are looking~for.
%\NM{Rewrote part of this section. Please check}

First, for any vector $\vec\gamma = (\gamma_i)_{i \in
  \set{n}}$, we define the \emph{$\vec\gamma$-counter graph}
$\bbC[\lambda^*, c_\src, \vec\gamma]$,
which is obtained from $\bbC[\lambda^*, c_\src]$ by replacing the cost
vector~$w$ on the edges with $\vec\gamma\cdot w$.
%where vertices are same as that
%of $\bbC[\lambda^*, c_\src]$, and each edge $\mathbf e = ((c,b), w,
%(c^{\prime},b^{\prime}))$ of $\bbC [\lambda^*, c_\src]$ is made to be
%$\mathbf e^{\vec\gamma} = ((c,b), \vec\gamma.w,
%(c^{\prime},b^{\prime}))$ in this $\bbC[\lambda^*, c_\src,
%  \vec\gamma]$.

We can then compute the cost of a $\vec\gamma$-minimal SPE by checking
existence of a path from $(c_\src, b^{c_\src})$ to $(c_\tgt, b)$ in
$\bbC[\lambda^*, c_\src, \vec\gamma]$, which minimizes the
$\vec\gamma$-weighted social cost.  Again, letting $\gamma_i =
1$ for all $i \in \set{n}$, a $\vec\gamma$-minimal SPE is a best SPE,
while taking $\gamma_i = -1$ for all $i \in \set{n}$, we get a
worst~SPE (maximizing social~cost).

We~can also solve the constrained-SPE-existence problem by
non-deterministically guessing an outcome and checking that it is a
path in~$\bbC[\lambda^{0^*},c_\src]$ and that it satisfies the
constraints. In~each case, we~can inductively build a strategy profile
witnessing the fact that the selected path is the outcome of an SPE.

\section{Conclusion and future works}
\label{sec:conclu}

In this paper, we introduced dynamic network congestion games, and
studied the complexity of various decision and computation problems
concerning social optima, Nash equilibria and subgame perfect equilibria.
Our~algorithms allow us to compute the price of anarchy and price of
stability for those games.

As future work,
%% we propose to explore stronger notions of equilibria,
%% and in particular subgame perfect equilibria~(SPEs)\SS{partially explored SPE now}, which better
%% suits the dynamic nature of our games by discarding non-credible
%% threats.  SPEs~for congestion games have been considered in the
%% context of Dynamic Resource Allocation Games
%% (DRAGs)~\cite{AHK-tcs20}. Reducing to DRAGs requires an \emph{a
%%   priori} upper bound on the number of steps to reach a target in
%% SPEs, so that it is unclear whether the results on SPEs in DRAGs
%% transfer to our dynamic~NCGs.
%
%A longer term research
our objective is to compute how
%measures such as
the price of anarchy and the price of stability (and~costs of equilibria and social optimum) evolve
when the number of players, seen as a parameter,~grows.
%\SS{Add some more stuffs on Conclusion?}
%, seen as a parameter.

% \bibliographystyle{alpha}
\bibliographystyle{plain}
\bibliography{newbib}

\newpage
\setcounter{theorem}{0}
\def\theHtheorem{\theHsection.\arabic{theorem}}
\def\theHlemma{\theHsection.\arabic{lemma}}
\def\theHproposition{\theHsection.\arabic{proposition}}
\def\theHcorollary{\theHsection.\arabic{corollary}}
\def\thetheorem{\thesection.\arabic{theorem}}
\def\thelemma{\thesection.\arabic{lemma}}
\def\theproposition{\thesection.\arabic{proposition}}
\def\thecorollary{\thesection.\arabic{corollary}}
\appendix
\section{Proofs of Section~\ref{section:socopt}}
\label{app-socopt}

\thmsocopt*

\begin{proof}
    A socially-optimal strategy can be obtained by computing a
    shortest path in~$\calP$ from~$\bar c_\src$ to~$\bar c_\tgt$.
    The~graph~$\calP$ has $O(n^{|V|})$ states, and shortest paths can be computed
    in polynomial time. We~can thus compute the (abstract) outcome of
    a socially-optimal strategy profile in exponential time. 
    We~easily derive a socially-optimal (blind) strategy profile.
    
    An abstract configuration can be stored in space $O(|V|
    \cdot \log(n))$, and deciding whether there is an edge of a given weight
    between two abstract configurations can be checked in~\PSPACE.
    Since shortest paths in~$\calP$ have length at most~$n\cdot|V|$,
    a~non-deterministic polynomial-space
    algorithm can guess a path step-by-step in~$\calP$,
    and check that it~reaches the target configuration with
    social-cost at~most~$b$.
    
%    and an abstract move vector can be stored in space~$O(|E|\times\log(n))$.
%    One can guess a path of size at most~$n^{|V|}$, which is the number of nodes of~$\calP$,
%    in polynomial space, and check that the cost of the path is less than~$b$.
%    One can use polynomial-space counters to store the length of the path and the total cost.
    By Savitch's theorem, this proves \PSPACE membership.

%
%    When~$n$ is fixed, a shortest path can be found in polynomial time in~$\calS$; hence,
%    the polynomial-time algorithms.
%\end{proof}

% \begin{restatable}{theorem}{thmsonphard}
%     \label{thm:so-np-hard}
%     The constrained social-optimum problem is \NP-hard in (acyclic)
%     dynamic NCGs.
% \end{restatable}
%\thmsonphard*

%\begin{proof}
    To prove the \NP-hardness we provide a reduction from the
    \textsf{Partition} problem which, given a family of integer
    numbers~$(r_i)_{1\leq i\leq m}$ that sum up to~$2S$, asks
    whether the family can be split into two subfamilies
    that both sum up to~$S$.  Consider an instance~$L$ of \textsf{Partition},
    and let
    %$S=\sum_{r_i \in L} r_i /2$ (which we may assume is an integer),
    $M = 14S + 12m + 1$, and $n=2S+2m$.  For~any $r\in\bbN$, we~define
    threshold cost function $\sfT_r$ as $\sfT_r(i)=1$ if~$i\leq r$,
    and $\sfT_r(i)=M$ otherwise
  %We can assume $\sum\limits_{r_i \in L} r_i$ is an even number (otherwise we can multiply all~$r_i$ by 2).\NM{This would be a different instance...}
%We construct a NCG~$\mathcal{G} = \tuple{\calA,k}$, with~$\calA = \tuple{V,E,\src,\tgt}$
%as follows.
We construct a dynamic NCG $\calA_L$ with
states~$V=\{\src,\tgt,d_1,d_2\} \cup \{s_i,a_{i,1},a_{i,2}\}_{i\in
  \set m}$.  The transitions are defined as follows:
\begin{itemize}
  \item for each~$i \in \set{m}$, there is a transition from~$\src$
    to~$s_i$ which cost function $\sfT_{r_i+2}$;
%assigns cost~$1$ up to $r_i+2$ players, and~$M$ for more.
  \item from each~$s_i$, there is a transition to~$a_{i,1}$ and
    another one to~$a_{i,2}$ with the same cost function that assigns
    cost $2$ for one player, and~$4$ for more;
  \item for~$i \in \set{m}$ and $j \in \{1,2\}$, there is a transition
    from~$a_{i,j}$ to~$d_j$ with constant cost~$1$, and a transition
    to~$\tgt$ with cost function $\sfT_{1}$;
%    which assigns cost~$1$ to one player, and~$M$ for more.
  \item from each~$d_j$, there is a transition to~$\tgt$ with
    cost function $\sfT_S$.
%    cost $1$ up to~$S$ players, and~$M$ for more.
\end{itemize}
The construction is illustrated in Figure~\ref{fig:reduction-SO}.

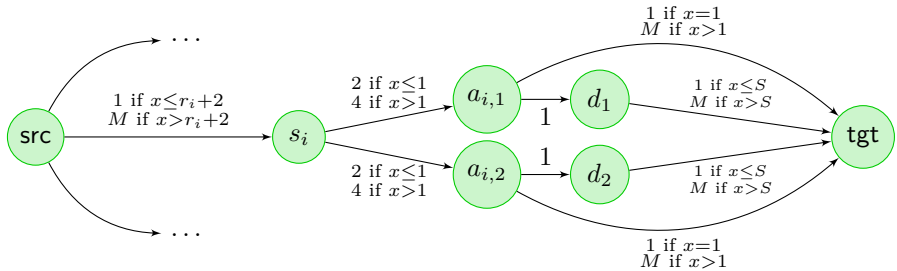
\begin{figure}[ht]
\begin{tikzpicture}[node distance=2.5cm]
    \begin{scope}
      \draw (0,0) node[rond,vert] (s) {\src};
      \draw (3.5,0) node[rond,vert] (s_i)  {$s_i$};
      \draw (2,1.3)  node (sitop) {$\ldots$};
      \draw (2,-1.3) node (sibot) {$\ldots$};
      \draw (6,0.5) node[rond,vert] (ai1) {$a_{i,1}$};
      \draw (6,-0.5) node[rond,vert] (ai2) {$a_{i,2}$};
      \draw (7.5,0.5) node[rond,vert] (d1) {$d_{1}$};
      \draw (7.5,-0.5) node[rond,vert] (d2) {$d_{2}$};

      \draw (11,0) node[rond,vert] (t) {\tgt};
      \path[-latex'] (s) edge node[above]
        { $\substack{1 \text{ if } x \leq r_i+2\\ M \text{ if } x>r_i+2}$}(s_i)
        (s_i) edge node[above]{ $\substack{2 \text{ if } x \leq 1\\ 4 \text{ if } x>1}$} (ai1)
        (s_i) edge node[below]{ $\substack{2 \text{ if } x \leq 1\\ 4 \text{ if } x>1}$} (ai2)
        (ai1) edge node[below]{1} (d1)
        (ai2) edge node[above]{1} (d2)
        (d1) edge node[above]{\footnotesize $\substack{1 \text{ if } x \leq S\\ M \text{ if } x>S}$} (t)
        (d2) edge node[below]{\footnotesize $\substack{1 \text{ if } x \leq S\\ M \text{ if } x>S}$} (t)
        (s) edge[bend left] (sitop)
        (s) edge[bend right] (sibot);
      \draw[-latex'] (ai1) edge[bend left=35] node[above]{ $\substack{1 \text{ if } x =1\\ M \text{ if } x>1}$} (t);
      \draw[-latex'] (ai2) edge[bend right=35] node[below]{ $\substack{1 \text{ if } x =1\\ M \text{ if } x>1}$} (t);
    \end{scope}
  \end{tikzpicture}
  \caption{The reduction of Theorem~\ref{thm:socopt}. %\ref{thm:so-np-hard}.
  The figure shows the edges involving states $s_i,a_{i,1},a_{i,2}$. The full graph is obtained
  by reproducing the shown construction for all~$ i \in \set m$.
  }
  \label{fig:reduction-SO}
\end{figure}

%, and
%$\vec{\lambda} = (\lambda_i)_{i \in \set n}$ with $\lambda_i =1$ for
%all $i \in \set n$ \NB{added $\vec\lambda$ here}.
%\NM{$k=2S+2m$?}
We prove that there exists a strategy profile in $\tuple{\calA,n}$
with
%$\vec{\lambda}$-weighted
social cost less than~$M$ if, and only~if,
$L$~is a positive instance of the partition problem.
%Recall that
%with the above-defined $\vec{\lambda}$, the $\vec{\lambda}$-weighted
%social cost coincides with the social cost.
%
Assume that the instance~$L$ has a solution $\tuple{L_1,L_2}$.  We describe
the behaviour of the strategy profile in~$\calA_L$ achieving social
cost less than~$M$.
\begin{itemize}
\item Initially, all players are at~$\src$. During the first step,
  for each~$i \in
  \set{m}$, $(r_i+2)$ players move to~$s_i$.  This incurs a total
  cost of $\sum_{i \in \set{m}}{r_i+2} = \sum_{i \in \set{m}} r_i
  + 2 \times m = 2S + 2m$;
    \item Consider~$i \in \set m$, and let~$j \in \{1,2\}$ be such
      that~$r_i \in L_j$.  From state~$s_i$ we let $r_i+1$ players
      move to $a_{i,j}$, and one player to $a_{i,3-j}$.  The $r_i+1$
      players each pay a cost of~$4$, and other player a cost of~$2$.
      Overall, we get a total cost of $\sum_{i \in \set{m}} ((r_i + 1)
      \times 4 + 1 \times 2) = 8S + 4m + 2m = 8S + 6m$;
    \item From each $a_{i,j}$, one player takes the transition
          directly to~$\tgt$.  All other players move to~$d_{j}$.
          Overall, $2m$ players directly move to $\tgt$, while
          all other players, that are $2S$ many of them, move to
          $\{d_1,d_2\}$.  The~total cost of this step is therefore
          $2S + 2m \times 2 = 2S + 4m$;
    \item From each $d_j$, players necessarily move directly to
      $\tgt$.  For each~$i \in \set{m}$ and $j\in\{1,2\}$, exactly
      $r_i$ players arrive to~$d_j$ from $a_{i,j}$.  Thus there are
      exactly $S$ players in each~$d_j$, so~that  the total cost of this step
      is~$2S$.
\end{itemize}
Summing over all steps, the social cost of this strategy profile is
$14S+12m < M$.

\smallskip Now for the opposite direction, consider any strategy
profile~$\pi$ in~$\tuple{\calA,n}$
with social cost less than~$M$. We are going to
construct a partition of~$L$.
%	In order to do that, we will show that, outcome of any strategy profile with cost $< opt$ is same as the outcome of the profile $\mathcal{P}_{L_1, L_2}$ above.
First we divide the path taken by players from $\src$ to~$\tgt$ in
three phases: \textbf{Phase~1:} from~$\src$ to some~$s_i$;
\textbf{Phase~2:} from some~$s_i$ to some~$a_{i,j}$; \textbf{Phase~3:}
from $a_{i,j}$ to~$\tgt$, either directly or via~$d_j$.
We now analyze the cost incurred by players under profile~$\pi$ in all
three phases.

In phase~$1$, each player pays either cost $1$ or cost $M$. By assumption on~$\pi$, all players
must have a cost of~$1$, and this is only possible if $r_i+2$ players move to~$s_i$, for each~$i \in \set{m}$.
The total cost of this phase is thus $2S+2m$.

In phase~$3$, a player pays a cost of either $2$, $M$ or $M+1$.  For
the latter two cases are not possible by assumption.  So phase~$3$'s
contribution to the social cost has to be $2 \times n = 4S + 4m$.
Then, the social cost of phase~$2$ is strictly less than
$(14S + 12m +1) - (2S +2m) - (4S+ 4m) = 8S + 6m +1$.  By~phase~$1$,
there are $r_i+2$ players at each $s_i$, so minimum contribution to
the social cost for these $r_i+2$ players is
$2 \times 1 + 4 \times (r_i+2) = 4r_i + 6$.  Summing
overall~$i \in \set{m}$, this yields
$\sum_{r_i \in L} (4r_i + 6) = 8S + 6m$.  Thus, the social cost
of phase~$2$ is $8S+6m$.  But this cost is achieved only when from
each $s_i$, one transition is taken by one player and the
other one by $r_i+1$ players.
    
Therefore, each $a_{i,j}$ contains either one player or $(r_i+1)$
players under~$\pi$.  Given the cost functions, at most~$2S$ players
can move to~$\tgt$ via~$d_1$ or~$d_2$ (otherwise players would get a
cost of~$M$).  So $2m$ remaining players must take the transition from
some~$a_{i,j}$ to~$\tgt$. But at each~$a_{i,j}$, there is a unique
player that takes this transition due to the cost function.
If~$a_{i,j}$ contains $r_i+1$ players, $r_i$ players take the path via
$d_j$, each paying a cost of~$2$. As~there are in total $2S$ players
taking the route via some~$d_j$, and~each~$d_j$ can contain at most
$S$ players because of the cost functions, there must be exactly $S$
many players which arrive at each~$d_j$ under~$\mathcal{P}$.  That is,
for each~$i,j$, there are exactly $r_i$ many players coming from some
$a_{i,j}$ to~$d_j$, and their total is~$S$.  This defines the required
partition of~$L$.
\end{proof}

\lemmashortshortest*

\begin{proof}
  We begin with proving that, for any path~$\rho$ in the
  multi-weighted graph~$\calM$ in which no player reaches the target
  state during the first $n$~steps, there is a path~$\rho'$ with social cost
  at most the social cost of~$\rho$ in which at least one player reaches the
  target state. This obviously extends to~$\calP$.

  To this aim, we consider, for each state~$s$ of~$\calA$, the
  shortest path from~$s$ to~$\tgt$ for a single player (i.e., when
  taking for each edge the value of the cost function when a single
  player takes that edge). Clearly enough, this is the best a player
  can hope from state~$s$. Now, we label all state~$s$ with the
  number~$\opt(s)$ of edges of the shortest (in terms of the number of
  edges) of the shortest (in terms of its single-player cost)
  paths. Then those values are at most $|V|-1$, and $\opt(s)=0$ if,
  and only~if, $s=\tgt$.

  Now, consider a path~$\rho$ in~$\calM$, in which we assume no player
  reaches~$\tgt$ during the first $|V|$ steps. We~build a table~$H$ of
  size $n\times |V|$, in which cell~$(i,j)$ contains
  $j+\opt(\rho(j)(i))$. Then for any~$i$, $H(i,1)$ is at most~$|V|$,
  and $H(i,|V|)$ is at least $|V|+1$, since we assume no players have
  reached~$\tgt$ after $|V|$ steps. Hence there is an index~$1\leq
  j_0<|V|$ that contains an element smaller than or equal to~$|V|$,
  and being the largest~so. Pick~$i_0$ such that $H(i_0,j_0)\leq |V|$,
  and consider the path~$\rho'$ in~$\calM$ obtained from~$\rho$ by
  making Player~$i_0$ follow an optimal path from
  $\rho(j_0)(i_0)$. Along this new path~$\rho'$, Player~$i_0$ will
  reach~$\tgt$ in $H(i_0,j_0)$ steps, which is less than or equal
  to~$|V|$.  Moreover, after step~$j_0$, Player~$i_0$ can be sure no
  other player takes the same edges, since this would mean that for
  some player~$i_1$, at some step~$j_1>j_0$, we~have $H(i_1,j_1)\leq
  |V|$. Hence the cost for Player~$i_0$ is at least as good
  along~$\rho'$ as it was along~$\rho$; also, the cost of all other
  players may only have decreased. It~follows that the social cost
  of~$\rho'$ is at least as good as that of~$\rho$.

  \smallskip
  Now, consider a shortest path in~$\calP$, and a corresponding
  path~$\rho$ in~$\calM$. Applying the arguments above, we~can build a
  path~$\rho'$ in~$\calM$ with social cost at least as good as that
  of~$\rho$, and in which at least one player reaches~\tgt. We~can
  then apply the same arguments recursively to the other players, each
  in a slide of at most $|V|$ steps. In~the end, we~get a path with
  minimal social cost and of length at most $n\times|V|$.  
\end{proof}

\section{Proofs of Section~\ref{section:nash}}
\label{app-nash}
\begin{restatable}{lemma}{lemmablindneexist}
    \label{lemma:blind-ne-exists}
    In every dynamic NCG, there exists a blind Nash equilibrium.
\end{restatable}
%\lemmablindneexist*
\begin{proof}
    Consider a dynamic NCG $\tuple{\calA,n}$ with $\calA=\tuple{V,E,\src,\tgt}$, 
    and a blind strategy profile~$\pi = (\sigma_i)_{i \in \set n}$.
    Recall that each blind strategy~$\sigma_i$ prescribes a path independently of the other players moves.
    Let us show that~$\psi$ is indeed a potential function,
    that is, for each profile~$\pi'$ obtained from~$\pi$ by modifying some player $k$'s strategy,
    we have
    \[
        \psi(\pi) - \psi(\pi') = \cost_{k}(\pi) - \cost_{k}(\pi').
    \]

    Consider such a profile~$\pi'$ obtained by replacing player~$k$'s strategy from~$\sigma_{k}$ to 
    some other blind strategy~$\sigma_{k}'$.
    Let us denote by~$e_1e_2\ldots$ the path of~$\sigma_{k}$, and by~$e_1'e_2'\ldots$ that of $\sigma_{k}'$.
    Let~$E'$ denote the set of edges of both paths.
    \begin{alignat*}1
        \psi(\pi) - \psi(\pi') &
        = \sum_{j=1}^{N_\pi}\sum_{e\in E} \sum_{i=1}^{\load_e(\pi,j)} \cost_e(i)
        -
        \sum_{j=1}^{N_{\pi'}}\sum_{e\in E} \sum_{i=1}^{\load_e(\pi',j)} \cost_e(i),
        \\
        &= \sum_{e\in E'} \sum_{j=1}^{N_\pi} \sum_{i=1}^{\load_e(\pi,j)} \cost_e(i)
        -
        \sum_{e\in E'} \sum_{j=1}^{N_{\pi'}} \sum_{i=1}^{\load_e(\pi',j)} \cost_e(i),
        \\
        &=\sum_{j=1}^{N_{\pi}}
            \cost_{e_j}(\load_{e_j(\pi,j)})
        - \sum_{j=1}^{N_{\pi'}}
        \cost_{e_j'}(\load_{e_j(\pi',j)})
        \\
        &=\cost_{k}(\pi) - \cost_{k}(\pi').
    \end{alignat*}
    Here, the second line is the definition; the third line follows from the fact that
    both terms are equal for edges~$e \not \in E'$ since there is no change on the load.
    The fourth line is due to the following observation. Assume that~$e_j \neq e_j'$.
    In this case, the load of~$e_j$ is decreased by~$1$ in~$\pi$, and that of~$e_j'$ 
    is increased by 1 in~$\pi'$.
    We have
    \[
        \sum_{i=1}^{\load_e(\pi,j)} \cost_e(i) - \sum_{i=1}^{\load_e(\pi',j)} \cost_e(i)
        =
        \cost_e(\load_e(\pi,j)) - \cost_e(\load_e(\pi',j));
    \]
    \enlargethispage{4mm}
    hence the result.
\end{proof}
\lemmablindnene*
\begin{proof}
    Consider an arbitrary blind strategy profile~$\pi=(\sigma_i)_{i
      \in\set n}$.  Assume that Player~$i$ has a profitable deviation,
    that is, there exists some (non-blind) strategy $\sigma_i'$ such that \(
    \cost_i(\pi[i \rightarrow \sigma_i']) < \cost_i(\pi)\). Define
    the blind strategy~$\sigma_i''$ as the strategy that follows
    Player-$i$'s path in the outcome of~$\pi[i \rightarrow
      \sigma_i']$. Then $\pi[i \rightarrow \sigma_i']$ and $\pi[i
      \rightarrow \sigma_i'']$
    %\SS{Do we need to use () for writing $\pi[i \rightarrow \sigma_i^{\prime}]$ and so on}
    have the same outcomes,
    %\NM{Only if we
    %  assume that the original strategy profile is blind (which we can
    %  do), right?},
    hence \( \cost_i(\pi[i \rightarrow
      \sigma_i'']) < \cost_i(\pi)  \).

    In other terms, if there is a profitable deviation strategy for player~$i$, then the new strategy
    can be chosen to be blind. This shows that a blind Nash equilibrium is a Nash equilibrium.
  \end{proof}
  
\begin{restatable}{lemma}{lemmabr}
    \label{lemma:br}
    Given a dynamic NCG~$\tuple{\calA,n}$, a blind strategy profile~$\pi$, and $i \in \set{n}$,
    Player~$i$ has a blind strategy that is their best response to~$\pi$.
    This strategy has size at most~$N_\pi+|V|$ and
    can be computed in time
    $O(|V|^2\cdot N_\pi^2)$.
    %\NM{I think it cound be $O((|V|\cdot N_\pi)^2)$.}
\end{restatable}
%\lemmabr*
\begin{proof}
  We let~$\calA=\tuple{V,E,\src,\tgt}$.
    We define a weighted graph~$\calG$ obtained by
    concatenating $N_\pi+1$ copies of~$\calA$, in which the moves of
    all other players (but Player~$i$) are hard-coded.
    Formally, $\calG=\tuple{V \times \set{N_\pi+1}, E'}$ where 
    \begin{itemize}
    \item for each edge $(v,\edgecost,v')$ in~$E$, there is an
      edge $((v,N_\pi+1),\edgecost(1),(v',N_\pi+1))$ in~$E'$: this
      encodes the fact that, after $N_\pi+1$ steps, all other players
      have reached the target state, and Player~$i$ plays alone;
    \item for each edge $e=(v,\edgecost,v')$ in~$E$ and each
      index~$1\leq k\leq N_\pi$, there is an edge $((v,k),w,(v',k+1))$
      in~$E'$ with $w=\edgecost(load_e(\pi_{-i},k)+1)$, where
      $\load_e(\pi_{-i},k)$ is the number of players
      (except~Player~$i$) taking edge~$e$ at step~$k$ when
      following~$\pi_{-i}$. This~way, $w$~corresponds to the cost
      incurred to Player~$i$ if they were to take edge~$e$ at
      step~$k$.
    \end{itemize}
    %%     in which all players but 
    %% Player~$i$ play follow their paths in the strategy profile in a synchronous way.
    %% Intuitively, $G$ contains~$N_\pi+1$ copies of~$\calA$, and the edge weights are Player-$i$'s cost.

    %% We define the weighted graph~$G=\tuple{C \times \set{N_\pi+1}, E'}$
    %% \NM{Do we need all configurations ($C$) in each copy, or only states are sufficient?}
    %% %, where~$E'$ contains the following edges.
    %% as follows:
    %% for all $c \in C$ and $\hat{m} \in M(c,i)$,
    %% $E'$~contains the following edges:
    %% \begin{itemize}
    %%     \item $((c,k),w(i),(c',k+1)) \in E'$ for all~$k \in \set{n}$,
    %%     \item $((c,N_\pi+1),w(i),(c',N_\pi+1))$,
    %% \end{itemize}
    %% where $T(c,(m_{1,k},\ldots, m_{i-1,k},\hat{m}, m_{i,k}, \ldots,m_{n,k})) = (w,c')$. Intuitively, this graph contains all possible strategies that Player~$i$ may play, with the associated cost when the other players follow their strategies.

%    Let~$\calT= \{ (\tgt,k) \mid 1\leq k\leq N_\pi+1\}$.
%
    By construction, any blind strategy~$\sigma_i'$ for Player~$i$
    in~$\tuple{\calA,n}$ corresponds to a path~$\rho'$ from
    $(\src,1)$ to~$(\tgt,N_\pi+1)$ in~$\calG$ such that the
    cost of~$\rho'$ is equal to $\cost_i(\sigma_{-i},\sigma_i')$:
    Player~$i$ can follow the exact same path as they would
    in~$\calS$ by ignoring the second component in the state space
    of~$\calG$.  Conversely, any path in~$\calG$ from~$(\src,1)$ to~$(\tgt,N_\pi+1)$
    corresponds to a path from~$\src$ to~$\tgt$ in~$\calA$, thus to
    some blind strategy~$\sigma_i'$. The cost of the path is equal to
    $\cost_i(\sigma_{-i},\sigma_i')$ by construction.
    
    This shows that the best-response strategy can be computed by a
    shortest path computation in~$\calG$.  This path has size at
    most~$N_\pi+|V|$: after~$N_\pi$ steps, the path enters
    configurations of the form $(c,N_\pi+1)$, so all other players
    have already reached the target; since the shortest path is
    acyclic, the bound follows.
    
    This computation can be done with Dijkstra's algorithm, which
    runs in $O((|V|\cdot N_\pi)^2)$.
\end{proof}

\thmbralg*
\begin{proof}
    We apply the previous lemma as follows.
    Consider an initial strategy profile~$\pi_0$ assigning to each player
    any acyclic path (thus, of length at most~$|V|$)
    from~$\src$ to~$\tgt$. We~have
    $\psi(\pi_0) %\leq(\cost_e(1) + \ldots + \cost_e(n))
    \leq |V|\cdot\sum_{e \in E}\sum_{i=1}^n\cost_e(i) \leq |V|n\max_{e \in E}\cost_e(n)$.
    This quantity is pseudo-polynomial in the size of~$\tuple{\calA,n}$.
%    and pseudopolynomial if~$n$ is fixed.\NM{even if $n$ is not fixed?}
    
%    Since each iteration takes polynomial time and reduces the value of~$\psi$ at least by one.\NM{weird sentence}
    Let~$\pi_1,\pi_2,\ldots$ denote the strategy profiles generated by this iterative procedure.
    Let~${m_i=N_{\pi_i}}$. By Lemma~\ref{lemma:br}, we have~$m_i \leq m_1 + (i-1)|V|$.
%\NM*{Update below depending on changes above}
    We~then have
    \begin{xalignat*}1
        \sum_{i=1}^k |V|^2m_i^2 &\leq |V|^{2}\sum_{i=1}^k (m_1+i|V|)^2\\
        & \leq |V|^{2}\sum_{i=1}^k (n|V|+i|V|)^2\\
        %          & = |V|^{4}\sum_{i=1}^k (i^2 + 2in + n^2)\\
        & \leq |V|^{4}\cdot k\cdot (n+k)^2
        %   %k^3 + kn + n^2k)
    \end{xalignat*}
    where the second step follows from~$m_1 \leq n|V|$.
    Applying~Lemma~\ref{lemma:br} again,
    the running time of the iterative procedure until the~$k$-th step is
    $O(|V|^4\cdot k\cdot (n+k)^2)$.
    In~the worst case, we~stop after $k=\psi(\pi_0)$ steps,
    so that the algorithm runs in pseudo-polynomial time.
    %% which is pseudo-polynomial.
    %% So the overall running time is also exponential.
    %% If~$n$ is fixed\NM{Also if not fixed?}, $\psi(\pi_0)$ is pseudopolynomial, so the algorithm 
    %% is also pseudopolynomial.
\end{proof}

\lemmablindsuboptimal*

\begin{proof}[Proof of Lemma~\ref{lemma:blind-suboptimal}]
%	There are better NEs when finite memory strategy is considered instead of considering path strategy only.
%	Here \textit{better NE} is better in terms of total cost of all the players.

	We consider the arena~$\calA$ shown in Fig.~\ref{fig:blind-suboptimal-app}
	with $n=3$ players, with $\src = q_0, \tgt = q_7$.
	\begin{figure}[htbp]
		\centering
		\begin{tikzpicture}[shorten >=1pt,node distance=1cm,on grid,auto, el/.style = {inner sep=3pt, align=left, sloped}]
		%% \node[rond,vert] (q0) at (0,0) {$q_0$};
		%% \node[rond,vert] (q1) at (1,2) {$q_1$};
		%% \node[rond,vert] (q2) at (3,2) {$q_2$};
		%% \node[rond,vert] (q3) at (5,2) {$q_3$};
		%% \node[rond,vert] (q4) at (1,-2) {$q_4$};
		%% \node[rond,vert] (q5) at (3,-2) {$q_5$};
		%% \node[rond,vert] (q6) at (5,-2) {$q_6$};
		%% \node[rond,vert] (q7) at (6,0) {$q_7$};
		\node[rond,vert] (q0) at (0,0) {$q_0$};
		\node[rond,vert] (q1) at (2,1) {$q_1$};
		\node[rond,vert] (q2) at (4,1) {$q_2$};
		\node[rond,vert] (q3) at (6,1) {$q_3$};
		\node[rond,vert] (q4) at (2,-1) {$q_4$};
		\node[rond,vert] (q5) at (4,-1) {$q_5$};
		\node[rond,vert] (q6) at (6,-1) {$q_6$};
		\node[rond,vert] (q7) at (8,0) {$q_7$};
		\begin{scope}[-latex']		
		\draw (q0) edge node[el, above] {$x\mapsto 2x$} node [el, below]{$e_1$} (q1);
		\draw (q0) edge node[el, above] {$x\mapsto 3x$} node [el, below]{$e_5$} (q4);
		\draw (q1) edge node[el, above] {$x\mapsto 3$} node [el, below]{$e_2$} (q2);
		\draw (q2) edge node[el, above] {$x\mapsto 3$} node [el, below]{$e_3$} (q3);
		\draw (q3) edge node[el, above] {$x\mapsto 2x$} node [el, below]{$e_4$} (q7);
		\draw (q4) edge node[el, above] {$x\mapsto x$} node [el, below]{$e_6$} (q5);
		\draw (q5) edge node[el, above] {$x\mapsto 2x$} node [el, below]{$e_7$} (q6);
		\draw (q6) edge node[el, above] {$x\mapsto 2x$} node [el, below]{$e_8$} (q7);
		\draw (q1) edge node[el, above] {$x\mapsto x$} node [el, below]{$\pun_1$} (q5);
		\draw (q2) edge node[el, above] {$x\mapsto 3$} node [el, below]{$\pun_2$} (q6);
		\end{scope}
                \end{tikzpicture}
		\caption{An arena on which blind Nash equilibria are sub-optimal.}
		\label{fig:blind-suboptimal-app}
	\end{figure}
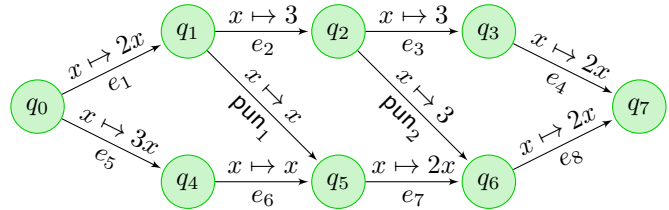

        	The strategy profile $\pi = (\sigma_i)_{i \in \set{3}}$ is  
	defined as follows.

	For~$i\in\{1,2\}$, 
	strategy $\sigma_i$ chooses $e_1e_2e_3e_4$ with the following exception:
	if Player~$3-i$ picks~$e_5$ at~$q_0$, then take~$\pun_1$ at~$q_1$;
	if Player~$3-i$ picks~$\pun_1$ at~$q_1$, then take~$\pun_2$ at~$q_2$.
	Strategy $\sigma_3$ follows $e_5e_6e_7e_8$.

	We have~$\cost_i(\pi) = 14$ for~$i \in \{1,2\}$, and~$\cost_3(\pi) = 8$,
	so~$\cost(\pi) = 36$.

	We first show that~$\pi$ is a Nash equilibrium. 
	
	% We then show that all
	% blind Nash equilibria have cost more than 
	% $\cost(\pi)$.
		First, Player $3$ has no incentive to deviate since taking $e_1$ at $q_0$ would alone cost $6$, and
		the rest of the path,
		either through $e_2$ or~$\pun_1e_7$,
		has a cost more than~$3$ so that the total cost
		is more than~$\cost_3(\pi)$.		
		For Players $1$ and $2$, there are three states where they can deviate. If Player~$i 
		\in \{1,2\}$ chooses $e_5$ at $q_0$, then 
		the cost is  $6+2+6+6 = 20$ as Player~$3-i$ chooses $\pun_1$.
		Similarly, if they choose $\pun_1$ at $q_1$ then Player~$3-i$ chooses $\pun_2$ at $q_2$ which yields
		a cost of $4+1+4+6 = 15 > \cost_i(\pi)$.
		Last, if Player~$i$ chooses
		$\pun_2$ at $q_2$, their cost is $4+3+3+4 = 14 = \cost_i(\pi)$. 
	%\end{proof}

	We now show that the profile~$\pi$ has a lower cost than any blind Nash equilibrium.

	There are only four possible blind strategies in our game following one of the four paths
	$\rho_1 = e_1 e_2 e_3 e_4, \rho_2 = e_5 e_6 e_7 e_8, \rho_3 = e_1 \pun_1 e_7 e_8$ and $\rho_4 = e_1 e_2 \pun_2 e_8$.
	We will represent these profiles as tuples of paths, \textit{e.g.} $(\rho_1,\rho_1,\rho_2)$.
	We are going to consider all possible tuples and show that each tuple is either not a Nash equilibrium or has cost more than
	$\cost(\pi) = 36$.
	
	Observe that when we permute the players in a Nash equilibrium, it remains a Nash equilibrium with the same social cost.
	So we only need to consider the case where all players use the same strategy ($4$ possibilities),
	the case where they all choose distinct strategies ($\binom{4}{3}$), 
	and the case where two of them choose the same strategy and the other one a different one ($\binom{4}{1} \times \binom{3}{1}$).
	So there are $20$ profiles to check, and the rest of the profiles are permutations of these.

	The following is an exhaustive list of the costs of these profiles.

        \noindent
        \begin{minipage}{.47\linewidth}
	\begin{align*}
	&\cost(\rho_1, \rho_1, \rho_1)= 18 \times 3 = 54 \\
	&\cost(\rho_2, \rho_2, \rho_2)= 24 \times 3= 72\\
	&\cost(\rho_3, \rho_3, \rho_3)= 21 \times 3 = 63\\
	&\cost(\rho_4, \rho_4, \rho_4)= 18 \times 3 = 54\\
	&\cost(\rho_1, \rho_2, \rho_3) = 12 + 13 + 12 = 37 \\
	&\cost(\rho_1, \rho_3, \rho_4) =  12+14+15=41\\
	&\cost(\rho_1, \rho_2, \rho_4) =  12+12+13=37\\
	&\cost(\rho_2, \rho_3, \rho_4) =  14+15+16=45\\
	&\cost(\rho_1, \rho_1, \rho_2) = 14\times 2 + 8 =\textbf{36} \\
	  &\cost(\rho_1, \rho_1, \rho_3)= 16\times 2 + 11 = 43\\
        \end{align*}
        \end{minipage}\hfill
        \begin{minipage}{.47\linewidth}
        \begin{align*}
	&\cost(\rho_1, \rho_1, \rho_4)= 14\times 2 + 14 = 42 \\
	&\cost(\rho_2, \rho_2, \rho_1)=16\times 2 + 10 = 42\\
	&\cost(\rho_2, \rho_2, \rho_3)=	20\times 2 + 15 = 55\\
	&\cost(\rho_2, \rho_2, \rho_4)= 18 \times 2 + 14 = 50 \\
	&\cost(\rho_3, \rho_3, \rho_1)= 16\times 2 + 14 = 46\\
	&\cost(\rho_3, \rho_3, \rho_2)= 18 \times 2 + 16 = 52\\
	&\cost(\rho_3, \rho_3, \rho_4)=  18 \times 2 + 15 = 51\\
	&\cost(\rho_4, \rho_4, \rho_1)= 16 \times 2 + 14= 46 \\
	&\cost(\rho_4, \rho_4, \rho_2)= 16\times 2 + 12 = 44\\
	&\cost(\rho_4, \rho_4, \rho_3)= 18 \times 2 + 15 = 51\\
	\end{align*}
	\end{minipage}
        
	All profiles have cost at least~$36$, and the only one that matches 36 is $(\rho_1,\rho_1,\rho_2)$.
	However, this is not a Nash equilibrium. In fact, the cost of Player~$1$ here is 14,
	but they could profit from deviating to~$\rho_3$ since
	$\cost_1(\rho_3,\rho_1,\rho_2)=13$.
\end{proof}

\section{Proofs of Section~\ref{section:spe}}
\label{app:spe}

\subparagraph{Equivalence between SPE and very-weak SPE~\cite{BBMR-csl15}.}

%We define \emph{weak}, and \emph{very weak} version of equilibria(Nash, and consequently Subgame-perfect) with two versions of single player deviation from a profile.
A
%Player-$i$
strategy $\sigma^{\prime}_i$ is said to be \emph{first-shot deviating} from another strategy~$\sigma_i$
if they coincide on all histories except the empty history.
%differs only at the initial history with respect to the applying step of the strategies.
%
A~strategy profile $\sigma = (\sigma_i)_{i \in \set n}$ is called a
\emph{very-weak Nash equilibrium} if for all $i \in \set n$, for every
first-shot deviating strategy $\sigma_i^{\prime}$, it holds
$\cost_i(\tuple{\sigma_{-i}, \sigma_i^{\prime}}) \geq
\cost_i(\sigma)$. Notice that this need not be an~NE~\cite{BBMR-csl15}.
Finally, a strategy profile is a \emph{very-weak subgame-perfect
  equilibrium} if it is a very-weak NE from every finite history of
the game.

We have the following equivalence:
\begin{restatable}{theorem}{weakequiv}
	\label{thm:weakequiv}
	In a dynamic NCG, a strategy profile is an SPE if, and
        only~if, it is a very-weak~SPE.
\end{restatable}

We introduce the following notations to ease the writing of the following proofs:
Given a strategy profile~$\sigma$ and a history~$h$, we~write
$\outcome(\sigma,h)$ for the outcome of the residual
strategy~$\sigma^h$ from the last configuration of~$h$.
Similarly, we~let $\cost_k(\sigma,h)=\cost_k(\outcome(\sigma,h))$.

\begin{proof}
  It~is easily checked that NEs are very-weak NEs, so that SPEs are very-weak SPEs.

%  Consider a dynamic NCG $\calG$.
%	A SPE strategy profile is clearly a very weak SPE too.
	
  For the opposite direction of the equivalence,
  let us consider a very-weak SPE $\sigma = (\sigma)_{i \in \set n}$ of a dynamic NCG~$\calG$.
  Suppose that $\sigma$ is not an~SPE; then
  there exists a history~$h$ of~$\cal{G}$ such that the residual strategy\ $\sigma^h$ is not an~NE.

  This implies that for some player~$i$, there exists a strategy $\sigma^{\prime}_i$ such that
  $\cost_i(\sigma,h) > \cost_i(\tuple{\sigma_{-i}, \sigma_i^{\prime}},h)$.
  In~particular, $\cost_i(\tuple{\sigma_{-i}, \sigma_i^{\prime}},{h})$~is finite, and along the outcome
  of this strategy profile, Player~$i$ reaches~$\tgt$.
  
  We pick $\sigma^{\prime}_i$ as a profitable deviating strategy (after history~$h$)
  with minimum number of deviations from~$\sigma_i$.
  Because Player~$i$ reaches~$\tgt$ along~$\outcome(\tuple{\sigma_{-i}, \sigma_i^{\prime}},{h})$,
  there exist profitable strategies with finitely many deviations.
  %	Now the strict inequality between the cost enforces $cost_i(\pi[i \rightarrow \sigma_i^{\prime}], h)$ to be finite.
%

  We~write~$h^{\prime}$ for the longest finite history along the path
  $\outcome(\tuple{\sigma_{-i}, \sigma^{\prime}_i},h)$ where
  $\sigma_i^{\prime}$ deviates from~$\sigma_i$, i.e,
  $\sigma_i(hh^{\prime}) \neq \sigma_i^{\prime}(hh^{\prime}) = c$ and
  $\outcome(\sigma,hh^{\prime}c) = \outcome(\tuple{\sigma_{-i},
    \sigma_i^{\prime}}, hh^{\prime}c)$.

%
%  We claim $\cost_i(\sigma, hh^{\prime}) > \cost_i(\tuple{\sigma_{-i},
%    \sigma^{\prime}_i}, hh^{\prime})$.
	Now if $\cost_i(\sigma, hh^{\prime}) \leq \cost_i(\tuple{\sigma_{-i},
    \sigma^{\prime}_i}, hh^{\prime})$ holds, then deviating at history $hh^{\prime}$ is unnecessary for Player $i$ for having profitable deviating strategy from $h$, and that contradicts the minimality of the number of deviations of $\sigma^{\prime}_i$ from $\sigma_i$.
Hence we obtain $\cost_i(\sigma, hh^{\prime}) > \cost_i(\tuple{\sigma_{-i}, \sigma^{\prime}_i}, hh^{\prime})$. 
%  it violates our hypothesis  deviations from~$\sigma$.
  But that implies we~have a first-shot profitable deviating strategy for 
%  \NM*{why single deviation?}
%  \SS*{can remove "single-deviating" if sounds not Ok, but it is just to emphasize we are talking about NEs}
%  \NM*{Does ``single'' refer to number of players, or number of positions where strategy differs? I think it should be the second case, and we need to argue why there is a single deviation...}
  Player $i$, $\sigma^{\prime}_i$, from history~$hh^{\prime}$,
  which contradicts to the assumption that $\pi$ is a very-weak~SPE.
%
%	Hence, a strategy profile is an SPE iff it is a very weak SPE.
%	
%	
\end{proof}

\subparagraph{Characterization of the outcomes of SPEs.}

%	\subparagraph{Region decomposition of the configuration graph.}

Consider a dynamic NCG~$\tuple{\calA,n}$, and the associated
configuration graph~$\calM=\tuple{C,T}$. We~partition the set~$C$ of
configurations into~$(X_j)_{0\leq j\leq n}$ such that a
configuration~$c$ is in~$X_j$ if, and only~if, $j=\#\{i\in\set n\mid
c(i)=\tgt\}$. Since~$\tgt$ is a sink state in~$\calA$, if~there is a
transition from some configuration in~$X_j$ to some configuration
in~$X_k$, then $k\geq j$.
We~define $X_{\geq j}=\bigcup_{i\geq j} X_i$, and
$\SE_{j}=\{(c,w,c')\in T \mid c\in X_{j}\}$ and $\SE_{\geq
  j}=\{(c,w,c')\in T \mid c\in X_{\geq j}\}$.

\medskip

%	\subparagraph{A family of labeling functions.}

%Let us recall
We~inductively define a sequence of families of labeling functions
$(\lambda^{j^*})_{0 \leq j \leq n}$, where for all~$0\leq j\leq n$,
$\lambda^{j^*} = (\lambda^{j^*}_i)_{i \in \set{n}}$, and for
all~$i\in\set n$, $\lambda^{j^*}_i\colon \SE_{\geq j} \rightarrow \bbN
\cup \{-\infty, +\infty\}$.

For any family $\lambda = \tuple{\lambda_i}_{i \in \set{n}}$ with
$\lambda^{j^*}_i\colon \SE_{\geq j} \to \bbN \cup \{-\infty, +\infty\}$
and any $c \in X_{\geq j}$, a~path $\rho =
(t_j)_{j \geq 1}$ from~$c$ visiting~$c_\tgt$ is said to be
$\lambda$-consistent if for any $i \in \set{n}$ and for any $1 \leq k
\leq |\rho|$, it~holds $\cost_i(\rho_{\geq k}) \leq \lambda_i(t_k)$.
We~write~$\LambdaSet_{\lambda}(c)$ to denote the set of all
$\lambda$-consistent paths from~$c$.

\smallskip
We now define the sequence $(\lambda^{j^*})_{0 \leq j \leq n}$.
For~$j = n$, we~have $X_{\geq n} = \{c_\tgt\}$ and
$\SE_{\geq n} = \{(c_\tgt,0^n,c_\tgt)\}$; there~is a single path,
which obviously is the outcome of an~SPE since no deviations are
possible. For~all~$i\in\set n$, we~let
$\lambda^{n^*}_i(c_\tgt,0^n,c_\tgt) =0$.
		
Now, fix~$j<n$, assuming that $\lambda^{(j+1)^*}$ has been defined.
In~order to define~$\lambda^{j^*}$, we~introduced an intermediary
sequence $(\mu^k_i)_{k\geq 0,i\in\set n}$ with $\mu^k_i\colon
\SE_{\geq j} \to \bbN\cup\{-\infty,+\infty\}$, of which
$(\lambda^{j^*}_i)_{i\in\set n}$ will be the limit.
We call these $\mu^k$'s as the \emph{$\lambda^{j^*}$-~building functions}.

Functions~$\mu^k_i$ mainly operate on~$\SE_j=\SE_{\geq j}\setminus
\SE_{\geq j+1}$: for~any~$\bfe\in \SE_{\geq j+1}$, we~let
$\mu^k_i(\bfe)=\lambda^{(j+1)^*}_i(\bfe)$.  Now,
for~$\bfe=(c,w,c')\in \SE_j$, the value $\mu^k_i(\bfe)$ is defined inductively
as follows:
	\begin{itemize}
		\item $\mu_i^0(\bfe)=0$ if $c(i)=\tgt$, and $\mu_i^0(\bfe)=+\infty$ otherwise;
		\item for~$k>0$, assuming~$\mu^{k-1}_i$ has been defined,
		we~let $\mu_i^k(\bfe)=0$ if~$c(i)=\tgt$, and
		\[
		\mu_i^k(\bfe)=
		\begin{cases}
          %	\min_{\substack{(c,w',c'')\in T\\ \forall j\not=i.\ c''(j)=c'(j)}}
               \mathop{\min\vphantom{\sup}}\limits_{c'' \in \dev_i(c,c')}\ 
		\sup\limits_{\rho\in\LambdaSet_{\mu^{l-1}}(c'')} ( \cost_i(c,c'') + \cost_i(\rho) ) & \quad\text{ if } \forall (c,\tilde{w}, \tilde{c}) \in T,~\LambdaSet_{\mu^{l-1}}(\tilde{c}) \neq \emptyset\\
		-\infty & \quad\text{ otherwise}
		\end{cases}
		\]  
	\end{itemize}

	\begin{lemma}\label{lemma-C4}
	    For any $\bfe \in T$, any $i \in \set{n}$, and any $k > 0$, we~have $\mu^k_i(\bfe) \leq \mu^{k-1}_i(\bfe)$.	
	\end{lemma} 
	
	\begin{proof}
%		Fix an $\mathbf e =(c, w, c^{\prime}) \in Z_{\geq j}$, and $i \in \set{n}$.
%		If $c^{\prime}(i) = \tgt$, then by definition $\mu^k_i(\mathbf e) = 0~\forall k \geq 0$, hence nothing further to prove.
        We prove the statement by induction on $k$, starting with~$k=1$.
        Write $\bfe = (c, w,c^{\prime})$.  If $c^{\prime} = \tgt$, then by definition $\mu^k_i(\bfe) = 0$  for all~$k$;
        otherwise, $\mu^0_i(\bfe) = \infty$; in~both cases, the result holds.
        
        Now, let us assume that for some~$k>0$, the following holds:
        %for all $0<l< k$,
        for all $i \in \set{n}$ and all $\bfe \in X_{\geq j}$, it~holds
        $\mu^l_i(\bfe) \leq \mu^{l-1}_i(\bfe)$.
        Then for any $c \in C$, it~implies $\LambdaSet_{\mu^k}(c,w,c^{\prime}) \subseteq \LambdaSet_{\mu^{k-1}}(c, w,c^{\prime})$.
        
        %Now consider $k = K+1$.
        
        Fix an arbitrary $\bfe = (c, w, c^{\prime}) \in T$. For~any
        $c^{\prime\prime} \in \dev_i(c,c^{\prime})$, by induction
        hypothesis, we have $\LambdaSet_{\mu^k}(c^{\prime\prime})
        \subseteq \LambdaSet_{\mu^{k-1}}(c^{\prime\prime})$.

       %Hence, 
        %\begin{align*}
        Then:
        \[
        \sup\limits_{\rho \in \LambdaSet_{\mu^k}(c^{\prime\prime})} (\cost_i(c,c^{\prime\prime}) +\cost_i(\rho))
          \leq \sup\limits_{\rho \in \LambdaSet_{\mu^{k-1}}(c^{\prime\prime})} (\cost_i(c,c^{\prime\prime}) + \cost_i(\rho)).
        \]
%
%        \implies ~\sup\limits_{\rho \in \LambdaSet_{\mu^k}(c^{\prime\prime})} (\cost_i(c,c^{\prime\prime}) + \cost_i(\rho)) &\leq \sup\limits_{\rho \in \LambdaSet_{\mu^{k-1}}(c^{\prime\prime})} (\cost_i(c,c^{\prime\prime}) + \cost_i(\rho))\\
        %       \implies
        It~follows
        \[
        \min\limits_{c^{\prime\prime} \in \dev_i(c,c^{\prime})}
       	\sup\limits_{\rho \in \LambdaSet_{\mu^k}(c^{\prime\prime})} (\cost_i(c,c^{\prime\prime}) + \cost_i(\rho)) \leq
        \min\limits_{c^{\prime\prime} \in \dev_i(c,c^{\prime})}
        \sup\limits_{\rho \in \LambdaSet_{\mu^{k-1}}(c^{\prime\prime})} (\cost_i(c,c^{\prime\prime}) + \cost_i(\rho)),
        \]
       	which shows $\mu^k(\bfe) \leq \mu^{k-1}(\bfe)$.
	\end{proof}

        \medskip

We~write $\LambdaSet^*$ for $\LambdaSet_{\lambda^{0^*}}$. Then:        
\lambdaconsistent*

\begin{proof}
  %	$(\Rightarrow:)$
%	 For the implication: $\Lambda^*(c_0) = \emptyset \implies \not\exists \text{ an SPE in } \calG$, we will show that if there exists an SPE then $\Lambda^*(c_0) \neq \emptyset$.
	
	Let us first consider an SPE $\sigma$ of~$\mathcal{G}$, of which~$\rho$ is the outcome. 
	We shall show that for any history $h$ ending in configuration~$c$,
        $\outcome(\sigma, h) \in \LambdaSet^*(c)$.
	In~fact, we~prove by induction on~$j$ that $\outcome(\sigma, h) \in \LambdaSet_{\lambda^{j^*}}(c)$ if $c \in X_{j}$.
	
%	We often use $last_h$ to be the last configuration of history $h$; in case of  $c_0 \xrightarrow{h} c$, we consider $c = last_h$.
%	Note that $\outcome(\sigma, h)$ is the path starting from $last_h$, followed by the unique play obtained by applying the profile $\sigma$ from $h$, and extending step by step by applying $\sigma$ until the target configuration.
	
%	Consider $c \in X_j$ for some $1 \leq j \leq n$.
%	We will show that $\outcome(\sigma, h) \in \Lambda^{j^*}(c)$.
%	In order to show the above for $c$, we prove that $\outcome(\sigma, h) \in \Lambda^{j^*}(c)$,assuming $\outcome(\sigma, h) = (c_r, w_r, c_{r+1})_{r \geq 1}$ which means $\forall i \in \set{n}~\cost_i()$
	
%	We will proceed by double induction: the first induction (backward) is on~$j$.
%	
%	\underline{Base case:}

	For $j = n$, for any history $h$ ending in~$X_n$, $\outcome(\sigma, h)$ is a path looping on~$c_\tgt$.
        We~defined  $\lambda^{n^*}(\bfe) = 0$ for $\bfe \in X_n$, so that the result holds.

        \medskip
	
	%\underline{Induction Hypothesis(IH):}
        Assume the statement is true for some $j+1 \leq n$: for~any history~$h$ ending in some configuration~$c\in X_{j+1}$,
        $\outcome(\sigma, h) \in \LambdaSet^{(j+1)^*}(c)$.
	%
	%\underline{Inductive step:}

	To~prove the result at level~$j$, we start another induction to show that for any~$k$, for any history~$h$
        whose last configuration~$c$ is in~$X_{j}$, the path $\outcome(\sigma,h)$ belongs to $\LambdaSet_{\mu^k}(c)$.
%=======
%	To~prove the result at level~$j$, we~start another induction to show that for any~$k$, for any history~$h$
%        whose last configuration~$c$ is in~$X_{j+1}$, the~path $\outcome(\sigma,h)$ belongs to $\LambdaSet_{\mu^k}(c)$.
%>>>>>>> 80bff0d18ed8bc26b0a43afb3980a0fc11d37d59
	
	%\emph{Base case: $k=0$}.
        \smallskip
        We~begin with the case $k=0$.
	Fix a history~$h$ ending in $c \in X_j$. Write $\outcome(\sigma, h) = (c_r, w_r, c_{r+1})_{r \geq 1}$.
	Let~$r^{\prime}$ be the least integer such that $c_{r^{\prime}}$ in $\outcome(\sigma,h)$ belongs to~$X_{>j}$.
	We~need to show that for any $i \in \set{n}$ and any~$r$,
        %if $c_{r}(i)=\tgt$, then
        $\cost_i((\outcome(\sigma, h))_{\geq r}) \leq \mu^0_i(c_r, w_r, c_{r+1})$.
        %for all  $p < r$.
        If~$r< r'$, then either $\mu^0(c_r, w_r, c_{r+1}) = \infty$, or $\mu^0(c_r,w_r,c_{r+1})=0$ and $c_r(i)=\tgt$;
          in~both cases, the result also holds.
        If $r\geq r'$, we~have $\mu^0_i(c_p, w_p, c_{p+1}) = \lambda^{(j+1)^*}_i(c_p, w_p, c_{p+1})$, and by our outer
        induction hypothesis, $\cost_i((\outcome(\sigma,h)_{\geq r})) \leq \lambda^{(j+1)^*}(c_r, w_r, c_{r+1})$.
        
%	For $p \leq \min\{r^{\prime}, r_i\}$, we know $\mu^0_i(c_r, w_r, c_{r+1}) = \infty$, so the result holds.
        %hence follows the required constraint.
%	If $r^{\prime} > r_i$, we are done.\NM{explain!}
%	Otherwise for $r_i > p > r^{\prime}$, by definition $\mu^0_i(c_p, w_p, c_{p+1}) = \lambda^{(j+1)^*}_i(c_p, w_p, c_{p+1})$,
%        and by IH of the outer induction we have $cost_i((\outcome(\sigma,h)_{\geq p})) \leq \lambda^{(j+1)^*}(c_p, w_p, c_{p+1})$ already.
%	That implies $\outcome(\sigma, h) \in \LambdaSet_{\mu^0}(c)$.

        \smallskip
	
%	\emph{Induction hypothesis:}
        We~now prove the induction step.
	We assume that $\outcome(\sigma, h) \in \LambdaSet_{\mu^{k-1}}(c)$ for any history $h$ ending in~$c\in X_j$.
        %\SS{Too many times ``assume $c$ and $h$" - need better phrasing}
%	
%	\emph{Inductive step}:
	%First we fix a $c \in X_j$, and as earlier let us denote
        Write $\outcome(\sigma, h) = (c_r, w_r, c_{r+1})_{r \geq 1}$.
        %with $c_\src \xRightarrow{h} c$.
	%
	We first observe that $\mu^k_i(c_r, w_r, c_{r+1}) \neq -\infty$ for any $r \geq 1$.
	This is because for any $(c_r, \tilde{w}, \tilde{c}) \in T$ and for any $c_\src \xRightarrow{\tilde{h}} \tilde{c}$, by our current induction hypothesis
        $\LambdaSet_{\mu^{k-1}}(\tilde{c})$ contains $\outcome(\sigma, \tilde{{h}})$, hence it~is not empty.
	
	Therefore, the only thing left to show now is that, for all $i \in \set{n}$, for all~$r$,
        we~have $\cost_i(\outcome(\sigma,h)_{\geq r}) \leq \mu^k_i(c_r, w_r, c_{r+1})$.
%where $c_{r_i}(i) = \tgt$.
        This is clear if $c_r(i)=\tgt$.  As~previously, we let
        $r^{\prime}$ be the least integer such that configuration
        $c_{r^{\prime}}$ along $\outcome(\sigma, h)$ belongs to $X_{>j}$.
	
%	For $r \leq \min\{r^{\prime}, r_i\}$, if $\cost_i(\outcome(\sigma,h)_{\geq r}) > \mu^k_i(c_r, w_r, c_{r+1})$ then,
	For $r < r^{\prime}$: assume $\cost_i(\outcome(\sigma,h)_{\geq r}) > \mu^k_i(c_r, w_r, c_{r+1})$; then,
        \[
	\cost_i(\outcome(\sigma,h)_{\geq r}) > %\mu^k_i(c_r, w_r, c_{r+1})  \\
%	&=\min\limits_{\substack{(c_r, w^{\prime}, c^{\prime}) \in T\\c^{\prime}(j) = c_{r+1}(j)\\\forall j \neq i}}
	\mathop{\min\vphantom{\sup}}\limits_{c'\in \dev_i(c_r, c_{r+1})} \
        \sup\limits_{\rho \in \LambdaSet_{\mu^{k-1}}(c^{\prime})} \{\cost_i(c_r,c^{\prime}) + \cost_i(\rho)\},
        \]
        which implies the existence of an edge $(c_r(i),\edgecost,v^{\prime})$ in~$E$ such that
        \[
        %\exists (c(i),\edgecost,v^{\prime}) \in E.\
        \cost_i(\outcome(\sigma,h)_{\geq r}) > \cost_i(c_r,c_{r+1}[i \rightarrow v^{\prime}]) +
           \cost_i(\sigma, h\cdot (c_r, w^{\prime}, c_{r+1}[i \rightarrow v^{\prime}])),
        \]
%	The last inequality uses the fact that
        because by our inner induction
        $\outcome(\sigma, h\cdot (c_r, w^{\prime}, c_{r+1}[i \rightarrow v^{\prime}])) \in
          \LambdaSet_{\mu^{k-1}}(c_r, w^{\prime}, c_{r+1}[i \rightarrow v^{\prime}])$.
        %, which is true by IH of the inner induction.
	  This shows a profitable first-shot deviation for Player~$i$,
          contradicting to our hypothesis that $\sigma$ is an SPE.
	
%	If $r_i \leq r^{\prime}$, the above case covers our argument.
	  Otherwise for $r \geq r^{\prime}$, the condition follows directly from the outer induction hypothesis,
          because $\mu^k$ coincides with~$\lambda^{j^*}$ for edges in those regions.

          \bigskip
          
	  %$(\Leftarrow:)$
          We~now prove the converse implication.
%          Suppose $\LambdaSet^*(c_{\src}) \neq \emptyset$, we will show that there exists an SPE.
%	
%	First of all, $\LambdaSet^*(c_{\src}) \neq \emptyset \implies \LambdaSet^*(c) \neq \emptyset$ for any $c_\src \xRightarrow{*} c$.
%	This is true because for any such $c \in X_j$, $\LambdaSet^*(c) = \emptyset$ comes from the fact $\LambdaSet_{\lambda^{j^*}}(c) = \emptyset$.
%	Then for $j-1$, $\mu^0_i(\mathbf e)$ would be $\bot$ for all $\mathbf e$ of the form $(c^{\prime}, w, \tilde{c})$ such that $c^{\prime} \Rightarrow c$, and this also makes $\LambdaSet_{\lambda^{(j-1)^*}}(c^{\prime})$ an $\emptyset$, and so on.
%	That is, the emptiness of $\LambdaSet_{\lambda^{j^*}}(c)$ gets propagated, and eventually makes $\LambdaSet^*(c_{\src})$ empty.
%	This is a contradiction by our hypothesis.
	
	%  In the next,
          Picking~$\rho\in \LambdaSet^*(c_{\src})$, 
          we build a very-weak SPE $\sigma = (\sigma_i)_{i \in \set{n}}$ (which is an SPE by Theorem \ref{thm:weakequiv}), step-by-step;
	  for~notational convenience, instead of defining $\sigma_i(h)$ for any history~$h$,
          we directly define what $\outcome(\sigma, h)$ would~be (if not already defined).
%	Remember, $\outcome(\sigma,h)$ is the path of the form $(c_j, w_j, c_{j+1}^{\prime})_{j \geq 1}$, where $c_1$ is the last configuration of $h$.\SS{$\outcome(\sigma, h)$ should be defined in prelims and then can remove this}
%	Obviously, to~maintain the consistency it follows : $\outcome(\sigma,h) = (c,w,c^{\prime}).\rho \implies \outcome(\sigma, h.(c, w, c^{\prime})) = \rho$.
%	A strategy profile in standard norm can be constructed from this way of defined profile.
	
%	First let us pick a $\rho \in \LambdaSet^*(c_\src)$ and we assign $\outcome(\sigma, c_{\src}) = \rho$.

	  As a first step, we let $\outcome(\sigma, c_{\src}) = \rho$.
          Now let us construct rest of the~$\sigma$.  We~pick an
          arbitrary non-initial history $h^{\prime} = h\cdot
          (c,w,c^{\prime})$ for which $\outcome(\sigma,h')$ is not
          defined; we~assume that $\outcome(\sigma, h)$ is already
          defined as $(c_j, w_j, c_{j+1})_{j \geq 1}$ such that
          $c^{\prime} \neq c_2$.
	
	If there exists $i_1 \neq i_2 \in \set{n}$ such that
        $c^{\prime}(i_j) \neq c_2(i_j)$ for both $j = 1$ and $j=2$, then pick any
        $\tilde{\rho} \in \LambdaSet^*(c^{\prime})$ (which we know is
        not~empty) and define $\outcome(\sigma, h\cdot (c,w,c^{\prime})) =
        \tilde{\rho}$.
	
	Otherwise $c_2$ differs from~$c'$ only at a single player's position, say Player~$i$'s.
	In~that case we define
	\begin{align}\label{def:sigma}
	\outcome(\sigma, h\cdot (c,w,c^{\prime})) = \mathop{\arg\max}\limits_{\rho \in \LambdaSet^*(c^{\prime})} \{\cost_i(c_1, c')+ \cost_i(\rho)\}
	\end{align}
%<<<<<<< HEAD
%	Note that, here we can take $\arg\max$ because $\lambda^*_i(\mathbf e)$ is finite (which we shall establish in corollary \ref{corr:finitenessoflambdastar}) for any $\mathbf e \in \mathcal{E}$, so there are finitely many plays in $\LambdaSet^*(c)$.
%	This ends our definition of the strategy profile $\sigma$.
%	
%	Now we will show that this $\sigma$ is indeed an SPE.
%	
%	Consider an arbitrary history $h$ and denote $\outcome(\sigma, h)$ by $(c_j, w_j, c_{j+1})_{1 \leq j \leq m}$ for some $m$.
%	Next we pick a configuration $c^{\prime}$ such that $(c_1, w^{\prime}, c^{\prime}) \in T$, $c^{\prime}(j) = c_2(j)~\forall j \neq i$ but $c^{\prime}(i) \neq c_2(i)$.
%	That is, $c^{\prime}$ is a configuration obtained by Player $i$'s deviation from $\sigma_i$ at the end of $h$.
%	
%	We will show that any such single player deviation is not beneficial if after that deviation Players continue to play with $\sigma$.
%	
%	Now,
%=======
	Note that, here we can take $\arg\max$ because
        $\lambda^*_i(\mathbf e)$ is finite (which we shall establish
        in Corollary \ref{cor:finitenessoflambdastar}) for any
        $\mathbf e \in \mathcal{E}$, so there are finitely many plays
        in $\LambdaSet^*(c)$.  This~ends our definition of the
        strategy profile~$\sigma$, which we now prove is an SPE.
	
%	Now we will show that this $\sigma$ is indeed an SPE.
	
	Consider an arbitrary history $h$ and denote $\outcome(\sigma, h)$ by~$(c_j, w_j, c_{j+1})_{1 \leq j}$.
        %for some $m$.
	Pick a configuration~$c^{\prime}$ such that $(c_1, w^{\prime}, c^{\prime}) \in T$, $c^{\prime}(j) = c_2(j)$ for all~$j \neq i$, but $c^{\prime}(i) \neq c_2(i)$.
	That is, $c^{\prime}$~is a configuration obtained by Player $i$'s deviation from~$\sigma_i$ at the end of~$h$.
	
	We show that  such a single player deviation is not beneficial if, after that deviation,
        all players continue to play following~$\sigma$:
%>>>>>>> 80bff0d18ed8bc26b0a43afb3980a0fc11d37d59
	\begin{align*}
	\cost_i(c_1,c^{\prime}) + \cost_i(\sigma, h\cdot(c_1, w^{\prime},c^{\prime})) 
	&= \sup\limits_{\rho \in \LambdaSet^*(c^{\prime})} \{\cost_i(\rho)+ \cost_i(c_1, c^{\prime})\} \tag{by \eqref{def:sigma}}\\
	&\geq \mathop{\min\vphantom{\sup}}\limits_{\tilde c\in \dev_i(c_1,c_2)}\
        %\min\limits_{\substack{(c_1, \tilde{w}, \tilde{c}) \in T\\\tilde{c}(j) = c_2(j)\\\forall j \neq i}}
         \sup\limits_{\rho \in \LambdaSet^*(\tilde{c})} \{cost_i(\rho)+ cost_i(c,\tilde{c})\}\\
	&= \lambda^*_i((c_1, w_1, c_2))\\
	&\geq \cost_i(\sigma, h)
	\end{align*}
	
	So no such configuration $c^{\prime}$ would be beneficial for Player $i$.
	Hence from any history $h$, $\sigma$ is very-weak~NE.
	Therefore, we can conclude that $\sigma$ is an~SPE.
\end{proof}

\subparagraph{Counter Graph.}
Having this correspondence between outcomes of SPE and $\LambdaSet^*$-paths,
our job to decide existence of SPE now reduces to checking
non-emptiness of~$\LambdaSet^*(c_\src)$.  To~this~aim, with any such
family $\mu = \tuple{\mu_i}_{i \in \set{n}}$ and configuration~$c$, we
associate an infinite-state
\emph{counter graph} $\bbC[\mu, c] = \tuple{C^{\prime},
  T^{\prime}}$ to capture $\mu$-consistent paths from~$c$:
\begin{itemize}
	\item the set of vertices
	  is~$C'= C\times (\bbN\cup\{+\infty\})^{\set n}$
          %\{(d,b)\in C\times (\bbN\cup\{+\infty\})^{\set n} \mid c
	%\Rightarrow^* d \text{ and } b\in [0, Y]\cup\{+\infty\}\}$
	%where~$Y= 6 |V|^n\cdot (\max_{e\in E} \edgecost_e(n))^{|V|}$;
%	\SS{todo: need to add the bounds}\NM{changed bound. Please check}
%	\SS{Should keep only the vertices reachable from $(c,b^c)$}
%        \NM{Why do we care if some states are not reachable?}
	\item $T'$~contains all edges
	$((d,b),w,(d',b'))$ such that $(d,w,d')$ is an edge of~$\calM$ and for
	  all~$i\in\set n$, either $b'(i)=0$ if~$d(i)=\tgt$, or
          $b'_i=\min\{b_i-w_i,
	\mu_i(d,w,d')-w_i\}$ otherwise (provided that~$b'_i\geq 0$ for
	all~$i$, in~order for~$(d',b')$ to be an edge of~$\bbC[\mu,c]$).
\end{itemize}
%
%\NM*{Actually, another way of presenting $\bbC[...]$ would be to define it as an infinite
%  graph, and show that only finitely many states are reachable. This is basically what prove by Thm C8-Lemma C10}
%
Intuitively, in configuration~$(d,b)$, $b$~is used to enforce
$\mu$-consistency: each edge taken along a path imposes a constraint
on the cost of the players for the rest of the path; this constraint
is added to the constraints of the earlier edges, and propagated along
the path.

With the initial configuration~$c$, we~associate $b^c$ such that
$b_i^c=0$ if~$c(i)=\tgt$ and $b_i^c=+\infty$ otherwise: this
configuration imposes no constraints since no edges have been
taken~yet.

Notice that $\bbC[\mu,c]$ is infinite, but as we show below, only
finitely many states are reachable from the initial
configuration. We~write $|C'|_r$ for the number of reachable states
in~$C'$.

We extend region decomposition of $\calM = \tuple{C,T}$ to any counter
graph $\bbC[\mu, c] = \tuple{C^{\prime}, T^{\prime}}$ in the natural
way: $(c^{\prime}, b^{\prime}) \in X^{\prime}_j$ if $c^{\prime} \in
X_j$, and an edge $((c^{\prime}, b^{\prime}), w^{\prime},
(c^{\prime\prime}, b^{\prime\prime})) \in \SE'_{\geq j}$ if
$(c^{\prime}, w^{\prime}, c^{\prime\prime}) \in \SE_{\geq j}$.

We call a path $\pi$ from $(c,b^c)$ to $(c_\tgt,b)$ (for~some~$b$) as a \emph{valid} path in $\bbC[\mu, c]$.
For any path $\pi \in \bbC[\mu, c]$, (where $c \in X_j$), we write
$\pi = \pi[j]\cdot \pi[j+1] \cdots \pi[n]$ where $\pi[l]$ denotes the
(possibly empty) section of path in~$\SE_l^{\prime}$.
%
%%%% HERE
%%%%We also use the notation $\maxb_l(\mu,c)$ \NM{need to restrict to reachable states}
%%%%(resp.~$\maxb_{\geq
%%%%  l}(\mu,c)$) to denote the maximum finite counter value in
%%%%$X^{\prime}_l$ (resp. $X^{\prime}_{\geq l}$) of $\bbC[\mu, c]$.  When
%
We also use the notation $\maxb_l(\mu,c)$ to denote the maximum finite counter value that appears in the vertices reachable from~$(c, b^c)$ and belonging
%\NM{what did you mean here?}
to $X_j^{\prime}$: more precisely,
\[
\maxb_l(\mu,c) = \max\{m \in \bbN\mid
\exists (c^{\prime},b^{\prime}) \in X_l \text{ s.t. }
(c, b^c) \rightarrow^{*} (c^{\prime}, b^{\prime})
\text{ and } \exists i \in \set{n}.\ b^{\prime}(i) = m\}.
\]
We extend this notation to
$\maxb_{\ge l}(\mu,c)$, which denotes $\max_{l \geq j} \maxb_l(\mu,c)$.
%% =======
%% We also use the notation $\maxb_l(\mu,c)$ to denote the maximum finite counter value that appear in the vertices reachable from $(c, b^c)$ and leong\NM{what did you mean here?} to $X_j^{\prime}$, more precisely $\maxb_l(\mu,c) = \max\{m \in \bbN: \exists (c^{\prime},b^{\prime}) \in X_l \text{ with } (c, b^c) \Rightarrow^{*} (c^{\prime}, b^{\prime}) \text{ in } \bbC[\mu,c] \text{ and } \exists i \in \set{n}~b^{\prime}(i) = m\}$ .
%% We naturally extend this notation to $\maxb_{\ge l}(\mu,c)$ to denote $\max_{l \geq j} \maxb_l(\mu,c)$.
%% >>>>>>> 5dacd462972c031889d829668ffe1961c9d13b92
%which (resp. $\maxb_{\geq
%  l}(\mu,c)$) to denote the maximum finite counter value in
%$X^{\prime}_l$ (resp. $X^{\prime}_{\geq l}$) of $\bbC[\mu, c]$.  
When
all the counter values are in~$\{0, \infty\}$, we~take
$\maxb_l(\mu,c) = 0$.

\begin{lemma}\label{lemma:lambdaconsis-pathsinCountGr}
	There exists a path $\pi = ((c_j, b_j), w_j (c_{j+1},
        b_{j+1}))_{j}$
        %from~$(c^{\prime}$
        from~$(c,b^c)$
        to a
        vertex $(c_\tgt, b)$ (for~some~$b$) in the counter
        graph $\bbC[\lambda,c] = \tuple{C^{\prime}, T^{\prime}}$ of~$\calG$
        if, and only~if, there is a $\lambda$-consistent
        path $\rho = (c_j, w_j, c_{j+1})_j$ from~$c$ in~$\calG$.
\end{lemma}

\begin{proof}
%		$(\Rightarrow:)$
  Assume that such a path~$\pi$ exists. Along~$\pi$, for each player~$i$, let~$k(i)$ be the least
  index such that $c_{k(i)}(i) = \tgt$.
  Then it is enough to show that for $1 \leq k < k(i)$, we~have
  $\cost_i(\rho_{\geq k}) \leq \lambda_i(c_k, w_k, c_{k+1})$.
  And indeed we have
  \begin{align*}
    0 \leq b_{k(i)}(i) &\leq b_{k(i)-1}(i) - \cost_i(c_{k(i)-1}, c_{k(i)})\\
    &\leq b_{k(i)-2}(i) - \cost_i(c_{k(i)-2}, c_{k(i)-1}) - \cost_i(c_{k(i)-1}, c_{k(i)})\\
    &\vdots\\
    &\leq b_{k+1}(i) - \sum\limits_{j = k+1}^{k(i)-1} \cost_i(c_j, c_{j+1})
  \end{align*}
  so that $ \cost_i(\rho_{\geq k}) \leq b_{k+1}(i) \leq \lambda_i(c_k, w_k, c_{k+1})$,
%			\implies \cost_i(\rho_{k}) \leq \lambda_i(c_k, w_k, c_{k+1})

  \medskip
  
%		$(\Leftarrow:)$
  Conversely, 
  if there is a $\lambda$-consistent path
  $\rho = (c_j, w_j, c_{j+1})_{1 \leq j < |\rho|}$ from~$c$,
  we~define $\pi = ((c_j, b_j), w_j, (c_{j+1}, b_{j+1}))_{1 \leq j < |\rho|}$ with:
  \begin{align*}
    b_1(i) &= \begin{cases} 0 & \text{ if  $c_1(i) = \tgt$} \\ +\infty & \text { otherwise }\end{cases} \\
    b_j(i) &= \min\{b_{j-1}(i) - \cost_i(c_{j-1}, c_j), \lambda_i(c_{j-1}, w_{j-1}, c_j)  - \cost_i(c_{j-1}, c_j)\} \quad\text{ for } 1 < j < |\rho|
  \end{align*}
  For~$\pi$ to be a valid path in $\bbC[\lambda,c]$, we have to prove that 
  $b_j(i) \geq 0$ for all $1<j\leq |\rho|$ and all $i \in \set{n}$.		
  For $j = 1$, it is evident.
  
  For $j > 1$, the second element of the minimum defining $b_j(i)$ is non-negative,
  because $\cost_i(c_{j-1}, c_j) \leq \cost_i(\rho_{\geq j-1}) \leq \lambda_i(c_{j-1}, w_{j-1}, c_{j})$.
  Suppose the first element $b_{j-1}(i) - \cost_i(c_{j-1}, c_j)$ is negative.
  Using definition of $b_{j-1}(i)$,
  we can say the above inequality can be true only if either
  $b_{j-2}(i) - \sum_{k = j-2}^{j-1}\cost_i(c_{k}, c_{k+1})  < 0$, or
  $\lambda_i(c_{j-2}, w_{j-2}, c_{j-1}) - \sum_{k = j-2}^{j-1}\cost_i(c_{k}, c_{k+1}) < 0$.
  But again the latter cannot be true
  because it directly contradicts $\lambda$-consistency of~$\rho$.
  Hence, our supposition can only be true if the former constraint holds.
  We~repeat the process with $b_{j-2}(i)$,
  coming down to the condition $b_2(i) - \sum_{k = 2}^{j-1} \cost_i(c_k, w_k, c_{k+1})<0$
  to make our supposition $b_{j-1}(i) - \cost_i(c_{j-1}, c_j)<0$ true.
  But $b_2(i) = \lambda_i(c_1, w_1, c_2) - \cost_i(c_1, c_2)$,
  so that the inequality above entails $\lambda_i(c_1, w_1, c_2) < \sum\limits_{k = 1}^{j-1} \cost_i(c_k, c_{k+1})$.
This contradicts $\lambda$-consistency condition at the beginning of~$\rho$.
It~follows that  $b_j(i) \geq 0$ for all $1 \leq j \leq |\rho|$ and  all $i \in \set{n}$.
\end{proof}

%\SS*{Further moved the lemma}
\begin{lemma}\label{lemma: charac-supremuminifinity}
	$\sup_{\rho \in \LambdaSet_{\mu^k}(c)} \cost_i(\rho) = +\infty$ if, and only~if,
	%for each of $(c, w^{\prime}, c^{\prime\prime}) \in T$ satisfying $c^{\prime\prime}(j) = c^{\prime}(j)~\forall j \neq i$,
	%  for all~$c''\in \dev_i(c,c')$,
	there exists a valid path~$\pi$ in~$\bbC[\mu^k, c]$ 
	%of the form $h.\beta.h^{\prime}$ where $\beta$ is a cycle.
	with the following conditions:
	\begin{itemize}
		\item $\pi$ is of the form $h\cdot \beta\cdot h^{\prime}$, where $\beta$ is a cycle in~$\bbC[\mu^k, c]$;
		%, i.e, denoting $\beta = ((c_m, b_m), w_m, (c_{m+1}, b_{m+1}))_{1 \leq m \leq |\beta|}$, we have $(c_1, b_1) = (c_{|\beta|+1}, b_{|\beta|+1})$, and
		\item Player $i$'s (constant) counter value~$b(i)$ is positive throughout~$\beta$.
                  %\SS{Should we write in singular form as counter values remain constant in a cycle?}
	\end{itemize} 
\end{lemma}

\begin{proof}
  %\SS{Does the proof need more elaboration?}
	%  By definition, that $\mu^k(c, w,c^{\prime}) = +\infty$ means that
	%  %for each of $(c, w^{\prime}, c^{\prime\prime})$ satisfying $c^{\prime\prime}(j) = c^{\prime}(j)~\forall j \neq i$
	%  for any~$c''\in\dev_i(c,c')$, it~holds
	%  $\sup_{\rho \in \LambdaSet_{\mu^{k-1}(c^{\prime\prime})}} \cost_i(\rho)=+\infty$.

	$\sup_{\rho \in \LambdaSet_{\mu^k}(c)} \cost_i(\rho) = +\infty$ implies there exists a sequence of paths $(\rho_m)_{m \geq 1}$ in $\LambdaSet_{\mu^k}(c)$ such that $\cost_i(\rho) \rightarrow \infty$ as $m$ grows.
	By Lemma \ref{lemma:lambdaconsis-pathsinCountGr}, for each of those $\rho_m$, there exists corresponding $\pi_m$ from $(c, b^{c})$ to $(c_\tgt, b)$ in $\bbC[\mu^k, c]$.
	As a player's cost in a single edge is bounded, the length of~$\pi_m$ has to grow unboundedly.
	First of all, that is only possible if there is a cycle $\beta$ in $\bbC[\mu^k, c]$ - we have the first condition now.
	%	Now note that, in a cycle of a counter graph, any player's counter value remains constant as counter values are by definition non-increasing.
	Now Player $i$'s counter value cannot be $0$ in $\beta$ because then the cycle doesn't contribute to make $\cost_i(\rho) \rightarrow 0$.
	Thus the second condition also holds true.

	Conversely, for $\pi = h\cdot \beta\cdot h^{\prime}$ in $\calC[\mu^k, c]$, where $\beta$ is a cycle and Player $i$'s cost $>0$ in $\beta$, we construct a sequence of paths $\pi^m = h\cdot \beta^{m+1}\cdot h^{\prime}$ for $m \geq 0$.
	By~Lemma~\ref{lemma:lambdaconsis-pathsinCountGr}, corresponding to this sequence of paths, there exists a sequence of  paths $(\rho_m)_{m \geq 0}$ in $\LambdaSet_{\mu^{k-1}}(c^{\prime\prime})$, and for those paths, $\cost_i(\rho_m) \rightarrow \infty$ as $m$~grows.
	Hence, $\mu^k_i(c,w,c^{\prime})=+\infty$ if these condition holds.
\end{proof}

\begin{lemma}\label{lemma:mu-stabilizes}
  For any edge $\bfe = (c, w, c^{\prime}) \in \SE_j$,
  the $(|V|+1)$-st family~$\mu^{|V|}$ of $\lambda^{j^*}$-building  functions
  satisfies $\mu^{|V|}_i(\bfe) \leq |V| \times \kappa$, where
  $\kappa=max_{e \in E}  \edgecost_e(n)$.
  %where $\mu^{|V|}$ is a $\lambda^{j^*}$-~building function.
\end{lemma}

\begin{proof}
	We prove a slightly stronger statement.
	Let~$m$ be the length of the shortest path from~$c(i)$ to~$\tgt$ in~$\calA$.
	We~show (by~induction) that $\mu^m_i(\mathbf e) \leq m \times \kappa$.
        %\max_{e \in E} \edgecost_e(n)$.
	As $m \leq |V|$ and by Lemma~\ref{lemma-C4}, this entails our current lemma.
	
%	We prove the stronger statement by induction on $m$.
	
	%\emph{Base case}: $m = 1$, i.e,
        If~$m=1$, then there is an edge~$e$ from $c[i]$ to~$\tgt$ in~$\calA$.
	Let us consider an edge $\bfe^{\prime} = (c, w^{\prime}, c^{\prime\prime})$
        such that $c^{\prime\prime}[j] = c^{\prime}[j]$ for all $j \neq i$, and $c^{\prime\prime}[i] = \tgt$.
	Now if $\LambdaSet_{\mu^{m-1}}(\tilde{c}) = \emptyset$ for any $(c, \tilde{w}, \tilde{c}) \in T$,
        then anyway $\mu_i^m(\mathbf e)=-\infty$, and the result holds.
	Otherwise, we~have $\LambdaSet_{\mu^{m-1}}(c^{\prime\prime}) \neq \emptyset$,
        and for any $\rho \in \LambdaSet_{\mu^{m-1}}(c^{\prime\prime})$, we~have $\cost_i(\rho) = 0$.
	Therefore, $\mu^m_i(\bfe) = \cost_i(c,c^{\prime\prime})$, which is bounded by~$\kappa$.
        %$\max_{e \in E} \edgecost_e(n)$.
	
        \medskip
	%\emph{Induction hypothesis}:
%        Let us assume for any configuration $c$, when the shortest distance from $c(i)$ to $\tgt$, $m$, it holds $1 \leq m \leq M-1$, then  $\mu^{m}_i(c,w,c^{\prime}) \leq m \times \max_{e \in E} \edgecost_e(n)$ for any $(c,w,c^{\prime}) \in T$ holds.	
%	\emph{Inductive step:}
        Now assume that the induction hypothesis holds up to step~$m-1$, and 
	consider a configuration~$c$ such that the length of a shortest path from~$c(i)$ to~$\tgt$ in~$\calA$ is~$m$.
        %, denoted by $m$, is $M$ and f
        Fix an edge $(c,w,c^{\prime}) \in T$.
	Write~$(c(i), \edgecost, v^{\prime}) \in E$ for the first edge of a shortest path from~$c(i)$ to~$\tgt$.
	We consider a configuration $c^{\prime\prime}$ such that $c^{\prime\prime}(j) = c^{\prime}(j)$ for all $j \neq i$,
        and $c^{\prime\prime}(i) = v^{\prime}$.
	By~construction, there is a path from $c^{\prime\prime}(i)$ to~$\tgt$ of length~$\leq m-1$.
	By~induction hypothesis, for any edge $(c^{\prime\prime}, \tilde{w}, \tilde{c}) \in T$,
        we~have $\mu^{m-1}_i((c^{\prime\prime}, \tilde{w}, \tilde{c})) \leq (m-1) \times \kappa$.
        %\max_{e \in E} \edgecost_e(n)$.
	This implies for any path $\rho = (t_j)_{j \geq 1} \in \LambdaSet_{\mu^{m-1}}(c^{\prime\prime})$,
        we~have $\cost_i(\rho) \leq \mu^{m-1}_i(t_1) \leq (m-1) \times \kappa$.
	Therefore, $\mu^m_i(c,w,c^{\prime}) \leq \cost_i(c,c^{\prime\prime}) + (m-1) \times \kappa \leq m \times \kappa$.
\end{proof}
	
%Recall that, for any $j$, the sequence $(\mu^k)_{k \geq 0}$ initializes from $\mu^0_i(c,w,c^{\prime}) = \infty$ for $(c,w,c^{\prime}) \in T$ with $c(i) \neq \tgt$

Lemmas~\ref{lemma-C4} and~\ref{lemma:mu-stabilizes} provide a bound
on the number of steps until any sequence of $\lambda^{j^*}$-computing
functions stabilize:
\begin{corollary}
  Any sequence~$(\mu^k)_{k\geq 0}$ of $\lambda^{j^*}$-computing
  functions stabilizes after at most $|V|(1+n\cdot\kappa\cdot |E|^n)$
  steps.
\end{corollary}

From Lemma~\ref{lemma:mu-stabilizes}, we also~get that 
the sequence $(\mu^k)_{k \geq 0}$ built for computing~$\lambda^{j^*}$
cannot stabilize unless for all $i \in \set{n}$,
for all $\mathbf e \in \SE_{\geq j}$, we~have $\mu^k_i(\mathbf e) \leq |V|
\times \kappa$.
%, where we recall that $\kappa=\max_{e \in E} \edgecost_e(n)$.
%, and when the sequence finally stabilizes we get $\lambda^{j^*}$.
%Hence,
By~Lemma~\ref{lemma-C4}:

\begin{corollary}\label{cor:finitenessoflambdastar}
  For any $0 \leq i, j \leq n$, and any $\bfe \in \SE_{\geq j}$, we~have
  $\lambda^{j*}_i(\bfe) \leq |V| \times \kappa$.
\end{corollary}

At this point, we have bounded the finite values that
$\lambda^{j^*}$'s can take. But~in the transitioning from
$\lambda^{(j+1)^*}$ to~$\lambda^{j^*}$, $\mu^k_i$~can take larger
values when $k < |V|$.  In~the sequel, we bound the values that any
family $\mu^{k} = \tuple{\mu^k_i}_{i \in \set{n}}$ can return.
	
%	By definition, finite values of $\mu^{k}_i(\bfe)$ mostly depend on the supremum of the Player $i$'s cost when taking supremum over all $\mu^{k-1}$-consistent paths from some designated configurations.
%	Hence, we first address in the following on given a $\mu$ what the bound of finite supremum of a player $i$'s cost in $\mu$-consistent paths would look like\SS{need better phrasing}\NM{indeed :-)}

To this aim, we~begin with working on the supremum of the cost for Player~$i$
of the paths in~$\LambdaSet_{\mu^k}(c)$.

When we consider a $\maxb_l(\mu^k,c)$ with $c \in X_j$ of $\calM$, it is implicit that $l \geq j$.
For $l > j$, we have
\[
\maxb_l(\mu^k,c) \leq \maxb_l(\mu^0,c) \leq 
  \max\limits_{\substack{\mathbf e \in \SE_l\\i \in \set{n}}}\lambda^{l^*}_i(\bfe)
  \leq
  |V| \times \kappa
  %\max_{e \in E} \edgecost_e(n)
\]
because $\mu^k_i(\bfe) = \mu^0_i(\bfe) = \lambda^{l^*}_i(\bfe)$
for those $\bfe \in \SE_{\geq l}$ and for all $i \in \set{n}$.
So it remains to bound $\maxb_j(\mu^k,c)$; but for that too,
when $k \geq |V|$, we have
$\maxb_j(\mu^k,c) \leq |V| \times \kappa$.
%\max_{e \in E} \edgecost_e(n)$.
Therefore, we only need to provide a bound for $\maxb_j(\mu^k,c)$
when $k < |V|$.

\begin{lemma}\label{lemma:BoundforMaxb}
	For an edge $\bfe = (c,w,c^{\prime}) \in Z_{j}$ and a
	$\lambda^{j^*}$-building function
	$\mu^k = \tuple{\mu^k_i}_{i \in \set{n}}$, if $\mu^k_i(\bfe)$ is non-zero finite then it~holds
	%	\[
	%		\maxb_j(\mu^k,c) \leq \{k \times (|C|+1) + |V|\} \times \max_{e \in E} \edgecost_e(n) ~\forall~k < |V|
	%	\]
	\[
	\mu^k_i(\bfe)\leq (n|C|+2|V|) \times
        %\{(n|C|)^{k-1} \kappa^k + \ldots + \kappa\}
        \sum_{l=1}^k (n|C|)^{l-1}\cdot \kappa^l
	%, \text{ where } \kappa = \max_{e \in E} \edgecost_e(n)
	\]
	Moreover, the above bound
        %works for
        also applies to $\maxb_j(\mu^k, c^{\prime\prime})$ for any $c^{\prime\prime} \in X_j$:
        \[
        \maxb_j(\mu^k, c) \leq (n|C|+2|V|) \times
        \sum_{l=1}^k (n|C|)^{l-1}\cdot \kappa^l.
        %\{(n|C|)^{k-1} \kappa^k + \ldots + \kappa\}$.
        \]
%	\SS{Do we need to recall $\kappa$ here?}
\end{lemma}

\begin{proof}
		We can claim that the finite maximum counter value appeared in $X_j^{\prime}$ of any counter graph $\bbC[\mu^k, c^{\prime\prime}]$ (where $c \in X_j$) is bounded by the finite maximum $\mu^k_i(\mathbf e)$ value appeared the same region (maximum over $i$ and $\mathbf e \in Z_j$), i.e,
		\[
		\maxb_j(\mu^k, c^{\prime\prime}) \leq
		\max\{\mu^k_i(c, w, c) \in \bbN \mid
		((c, b), w, (c^{\prime}, b^{\prime}))
		\in \SE_{j}^{\prime} \text{ of } \bbC[\mu^k, c^{\prime\prime}],\ i \in \set{n}\}.
		\]
		This is justified because for the initial vertex $(c,b^c)$ of $X_j^{\prime}$, $b^c \in \{0,\infty\}$.
		Hence, a counter value becomes non-zero finite, when some $\mu^k_i(\mathbf e)$ for $\mathbf e \in \SE_j$ becomes non-zero finite.
		But by definition of counter graph, $((c, b), w, (c^{\prime}, b^{\prime})) \in \SE_{j}^{\prime}$ implies $(c, w, c^{\prime}) \in \SE_j$, hence we can obtain,
		\[
		\maxb_j(\mu^k, c^{\prime\prime}) \leq \max\{\mu^k_i(c, w, c^{\prime}) \in \bbN \mid (c,w,c^{\prime}) \in \SE_j, i \in \set{n}\}
		\]
%	Thus it is enough to  bound  $\mu^k_i(\bfe)$ only.
%	We proceed by induction on~$k$.
%	In~fact, we consider both the statements together in a single statement to be proven inside one induction as it follows below:
%
%        \emph{Base case:} $k = 0$.
%	As $\mu^0_i(\mathbf e) \in \{0,\infty\}$ for all $\mathbf \in \SE_j$, the given bounds hold for $k = 0$.
%	
%	
%	\emph{Induction Hypothesis:}
%	Suppose for $k \leq K$ for some $K < |V|$, it holds,
        By~induction on~$k$, we prove that
	for any $c^{\prime\prime} \in X_j$ and any~$i\in\set n$,                 
        \[
	\maxb_j(\mu^k, c^{\prime\prime}) \leq \max\{\mu^k_i(\bfe) \in \bbN \mid \bfe \in Z_j\} \leq(n|C|+2|V|) \times
        \sum_{l=1}^k (n|C|)^{l-1}\cdot \kappa^l.
        %\{(n|C|)^{k-1} \kappa^k + \ldots + \kappa\}
        %~\forall i \in \set{n}
	\]

	As $\mu^0_i(\mathbf e) \in \{0,\infty\}$ for all $\mathbf \in \SE_j$, the given bounds hold for $k = 0$.
        
	We fix an arbitrary $\bfe  = (c, w,c^{\prime})\in \SE_j$ and a player $i \in \set{n}$ such that $\mu^k_i(\bfe) \in \bbN \setminus \{0\}$.
	That $\mu_i^k(\bfe)$ is a non-zero finite value means that
	$\sup_{\rho \in \LambdaSet_{\mu^{k-1}(\tilde{c})}} \cost_i(\rho) \in \bbN$
	for some $\tilde{c} \in \dev_i(c,c')$.	
	%	When for an $\mathbf e = ((c^{\prime}, b^{\prime}), w^{\prime}, (c^{\prime\prime}, b^{\prime\prime})) \in Y_j$\SS{both edges are denoted by $\mathbf e$?}, $\mu^k_i(c^{\prime}, w^{\prime}, c^{\prime\prime})$ is a non-zero finite value, we get $\sup\limits_{\rho \in \LambdaSet_{\mu^{k-1}}(\tilde{c})} \cost_i(\rho) \in \bbN$ for some $\tilde{c}$ such that $\tilde{c}(j) = c^{\prime\prime}(j)~\forall j \in \set{n}$.
	If $\tilde{c}(i) = \tgt$ then $\sup_{\rho \in \LambdaSet_{\mu^{k-1}}(\tilde{c})} \cost_i(\rho) = 0$, and that makes $\mu^k_i(\mathbf e) \leq \cost_i(c, \tilde{c}) \leq \kappa$, satisfying the given bound.
	
	Otherwise, $\tilde{c}(i) \neq \tgt$,
	and depending whether $\tilde{c}$ belongs to $X_j$ or~$X_{>j}$, we analyze two cases:
	\begin{itemize}	
		%	\emph{Case 1:}
		\item $\tilde{c} \in X_l$ for some $l > j$:
		Then
		%	\begin{align*}
		\begin{xalignat*}1
		\mu^k_i(\mathbf e) &\leq \sup\limits_{\rho \in \LambdaSet_{\mu^{k-1}}(\tilde{c})} \{\cost_i(c, \tilde{c}) + \cost_i(\rho)\} \\
		&\leq \cost_i(c, \tilde{c}) + \max\limits_{\substack{\mathbf e^{\prime} \in \SE_l\\i \in \set{n}}} \mu^{k-1}_i(\mathbf e^{\prime})
		\leq (1+|V|) \times \kappa.
                \end{xalignat*}
		
		%	\end{align*}
		
		%	\emph{Case 2:}
		\item $\tilde{c} \in X_j$:
		Now consider $\bbC[\mu^{k-1}, \tilde{c}]$, and its initial vertex $(\tilde{c}, b^{\tilde{c}})$.
		By Lemma \ref{lemma:lambdaconsis-pathsinCountGr}, for any path $\rho \in \LambdaSet_{\mu^{k-1}}(\tilde{c})$, we have a corresponding valid path $\pi_\rho$ in $\bbC[\mu^{k-1}, \tilde{c}]$.
		We also consider the region decomposition $\pi_\rho = \pi_\rho[j] \ldots \pi_\rho[n]$.
		From the fact that ${\sup_{\rho \in \LambdaSet_{\mu^{k-1}}(\tilde{c})} < +\infty}$, we~argue that in~$\pi_\rho$, from $(\tilde{c}, b^{\tilde{c}})$ within each $|C|$ step either one counter value strictly decreases, or Player $i$ reaches the $\tgt$.
		Otherwise, there would have been a cycle in $\bbC[\mu^{k-1}, \tilde{c}]$ resulting the supremum $+\infty$ ( thanks to lemma \ref{lemma: charac-supremuminifinity}).
		Recall by design, counter values of $\pi_\rho$ in $X_j^{\prime}$ lie in $\set{\maxb_j(\mu^{k-1}, \tilde{c})} \cup \{0, +\infty\}$, and when a counter value decreases along an edge, it decreases at least by $1$.
		Therefore, within $n \times |C| \times (\maxb_j(\mu^{k-1}, \tilde{c})+1)$ steps from $(\tilde{c}, b^{\tilde{c}})$ at least one of the counter value becomes $0$.
		When a player-$l$ counter value becomes $0$, $\pi_\rho$ must reaches the next region (making the corresponding $c(l) = \tgt$), otherwise $\pi_\rho$ is not a valid path.
		Moreover, from the next region, $\cost_i(\rho[j+1] \ldots \rho[n])$ is bounded by $\mu^{k-1}_i(\bfe^{\prime})$, where $\bfe^{\prime}$ denotes the first edge of $\rho$ which belongs to $\SE_{>j}$. 
		Therefore,
		\begin{align*}
		\cost_i(\rho) 
		&\leq (n \times |C| \times (\maxb_j(\mu^{k-1}, \tilde{c})+1) )\times \kappa + \max_{l > j}\maxb_l(\mu^{k-1}, \tilde{c})\\
		&\leq (n \times |C| + |V|) \times \kappa + (n \times |C|) \times \maxb_j(\mu^{k-1}, \tilde{c}) \times \kappa.
		\end{align*}
		If Player $i$ reaches $\tgt$ within $X_j^{\prime}$, $\cost_i(\rho)$ would be much smaller.
		In the above, we have used the bound from corollary \ref{cor:finitenessoflambdastar} as $\mu^{k-1}(\bfe) = \lambda^{l^*}(\bfe)$ for $\bfe \in X_l$ for $l > j$.
		Now as $\tilde{c} \in X_j$, we can use induction hypothesis for giving bound to $\maxb_j(\mu^{k-1}, \tilde{c})$.
		As the above shown bound works for any $\rho \in \LambdaSet_{\mu^{k-1}}(\tilde{c})$, it works for the supremum too, hence we have
		\begin{align*}
		\mu^k_i(c, w, c^{\prime}) &\leq \sup\limits_{\rho \in \LambdaSet_{\mu^{k-1}}(\tilde{c})} \cost_i(c^{\prime}, \tilde{c}) + \cost_i(\rho)\\
		&\leq |V| \times \kappa + (n \times |C| + |V|) \times \kappa + (n \times |C|) \times \maxb_j(\mu^{k-1}, \tilde{c}) \times \kappa\\
		&\leq (n|C| + 2|V|) \times \kappa + (n|C|\times \kappa) \times
                \{(n|C|+2|V|) \times
                %((n|C|)^{k-1} \kappa^{k-1} + \ldots +\kappa)\}
        \sum_{l=1}^{k-1} (n|C|)^{l-1}\cdot \kappa^l                
                \\
		&= (n|C|+2|V|) \times
                \sum_{l=1}^{k} (n|C|)^{l-1}\cdot \kappa^l
                %\{(n|C|)^{k-1} \kappa^k + \ldots + \kappa\}
		\end{align*}
		%		then we claim that from for any valid path $\rho = ((c_m, b_m), w_m, (c_{m+1}, b_{m+1}))_{m \geq 1}$ in $\bbC[\mu^{k-1}, \tilde{c}]$, at least one of the following holds true:
		%		\begin{itemize}
		%			\item There exists $i^{\prime} \in \set{n}$ such that $b_1(i^{\prime}) = \infty$ and $b_{|C|+1}(i^{\prime}) \in \bbN$
		%			
		%			\item $c_{|C|+1}(i) = \tgt$
		%		\end{itemize}
	\end{itemize}
\end{proof}

%This ends providing bound for any $\maxb_l(\mu^k,c)$ where $c \in X_j$, $n \geq l > j$ and $\mu^k$ taken with respect to~$j$.
%Now substituting $(\kappa^k + \ldots + \kappa)$ by $k\cdot\kappa^k$ and substituting these bounds in Theorem~\ref{thm:finboundforSup}, we obtain:
%for any configuration $c \in X_j$ and any $\lambda^{j^*}$-building function
%$\mu^k = \tuple{\mu^k_i}_{i \in \set{n}}$,
%%where $\mu^k$ is with respect to $j$,
%if $\sup_{\rho \in \LambdaSet_{\mu^{k}}(c)} \cost_i(\rho) \neq +\infty$, then
%\begin{align*}
%	\sup_{\rho \in \LambdaSet_{\mu^{k}}(c)} \cost_i(\rho) \leq 
%	\begin{cases}
%	(n-j+1)\times|C|\times \kappa + (n-j) \times 2.|V|.\kappa^2 + (|C| + 2.|V|) \times k \times \kappa^{k+1} &\text{ if } k < |V|\\
%		|V| \times \max_{e \in E} \edgecost_e(n) &\text{ otherwise}
%	\end{cases}
%\end{align*}
%where $\kappa = \max_{e \in E}\edgecost_e(n)$.
This ends having bound for any $\maxb_l(\mu, c)$ for any $\lambda^{j^*}$-building function $\mu$, $c \in X_j$, and $l \geq j$.
Hence, we can conclude that the counter graph $\bbC[\mu,c] = \tuple{C^{\prime}, T^{\prime}}$ can be made finite, by taking
$C^{\prime} = \{(d,b) \in C \times ([0;Y] \cup \{+\infty\})^{\set{n}} \mid
c \Rightarrow^* d\}$,
with $Y=n^{|V|}\cdot |V|^{n.|V|} \cdot \kappa^{|V|}$, which is doubly-exponential in the encoding of $n$.
Note that, this makes $|C'|_r$ at most double-exponential too - which will be one of the key arguments in the complexity analysis of the final algorithm.

%we finally have
%\begin{corollary}
%	$Y$ is in $\calO(|V|^{\set{n}}.(\max_{e \in E} \edgecost_e(n))^{|V|})$, i.e, double-exponential in the input size
%\end{corollary}
%\NM*{I think $Y\leq 6\cdot n |V|^{n+1} \cdot \kappa^{|V|}$ (more precise than $O(...)$).}

\subparagraph{Algorithm.}
Now we describe in details how we check non-emptiness of $\LambdaSet_{\lambda^*}(c_\src)$.
Starting from~$\lambda^{n^*}$, we inductively compute $\lambda^{j^*} = \tuple{\lambda^{j^*}_i}_{i \in \set{n}}$.
When computing~$\lambda^{j^*}$, we use an intermediary induction (following the definition): we initialize with the definition of $\mu^0 = \tuple{\mu^0_i}_{i \in \set{n}}$.
%As the induction hypothesis, we assume
Then, assuming we have already computed $\{\mu^{k-1}_i(\mathbf e)\mid \mathbf e \in T\}$,  a general step to compute $\mu^k_i(c, w, c^{\prime})$ is as follows:

\begin{itemize}
\item 
First we need to check that $\mu^k_i(c, w,c^{\prime})$ is $-\infty$.
For that, we check, one by one, whether for each of $(c, w^{\prime}, c^{\prime\prime}) \in T$, $\LambdaSet^{\mu^{k-1}}(c^{\prime\prime})$ is not $\emptyset$.
In order to do that,we guess a valid path (but do not store) in $\bbC[\mu^{k-1}, c^{\prime\prime}]$.
If there is such a path in $\bbC[\mu^{k-1}, c^{\prime\prime}]$, there will be one, length of which is bounded by $|C|$, which is doubly-exponential in the size of input.
Hence we keep a counter which we increase by~$1$ every time we correctly guess a new edge along the path, and the counter stops either when we reach a $(c_\tgt, b)$, or at the latest when it crosses $|C|$ marks - this counter requires exponential space to encode.
%
%%%%
%%%%Hence we keep a counter which we increase by $1$ every time we correctly guess a new edge along the path, and the counter stops when it crosses $|C|$ marks - this counter requires \EXPSPACE to encode.
%%%%We keep a binary counter initialized at $0$, and gets value $1$ once we reach a vertex $(c_\tgt, b)$ (for some valid $b$).
%

\item if the previous check failed,
  %Once we have $\LambdaSet_{\mu^{k-1}}(c^{\prime\prime})$ non-empty for all $(c, w^{\prime}, c^{\prime\prime}) \in T$,
  we know that $\mu^k_i(c, w, c^{\prime})$ would be in $\bbN \cup \{+\infty\}$.
%Our next job is to compute the values.

%Here we guess $\mu_i^k(c, w, c^{\prime}) = M$ 
%To compute $\mu^{k}_i(c,w,c^{\prime})$, we guess the value (each finite and $+\infty$)
%and then verify it-
%\begin{itemize}
%	\item 
%          First we guess that $M$ to be $+\infty$.
  We now check if $\mu^{k}_i(c,w,c^{\prime})=+\infty$.
  For that we need to verify the conditions stated in Lemma~\ref{lemma: charac-supremuminifinity}, i.e, for each $c^{\prime\prime} \in \dev_i(c,c^{\prime})$, we guess a valid path $\pi$ of the form $h\cdot \beta\cdot h^{\prime}$ in $\bbC[\mu^{k-1}, c^{\prime\prime}]$, where $\beta$ is a cycle and Player $i$'s counter value is $>0$ inside~$\beta$.
%%%% 
%%%%          First we guess that $\mu^k_i(c, w,c^{\prime})$ to be $+\infty$.
%%%%	For that we need to verify the conditions stated in Lemma \ref{lemma: charac-supremuminifinity}, i.e, for each $(c, w^{\prime}, c^{\prime\prime}) \in T$ satisfying $c^{\prime\prime}(j) = c^{\prime}(j)~\forall j \neq i$, we guess a path $\pi$ of the form $h.\beta.h^{\prime}$ from $(c^{\prime\prime}, b^{c^{\prime\prime}})$ to $(c_\tgt, b)$ in $\bbC[\mu^{k-1}, c^{\prime\prime}]$.
%	For each $(c, w^{\prime}, c^{\prime\prime}) \in T$ such that $c^{\prime\prime}(j) = c^{\prime\prime}(j)~\forall j \neq i$, we claim the following: $\mu^k(c,w,c^{\prime})$ is $+\infty$ iff there is a path $\pi$ in $\LambdaSet_{\mu^k}(c^{\prime})$ from $(c^{\prime}, b^{c^{\prime}})$ to $(c_\tgt, b)$ such that $\pi$ is of the form $h.\beta .h^{\prime}$, where denoting $\beta = ((c_m,b_m), w_m, (c_{m+1}, b_{m+1}))_{1 \leq m \leq |\beta|}$, we have $(c_{|\beta|+1}, b_{|\beta|+1}) = (c_1, b_1)$ and $b_1(i) = \infty$.\SS{Probably would be better if we state it separately}

%%% 
          We first guess the first vertex of~$\beta$, say  $(c_1, b_1)$ (with $b_1(i)>0$), from which the cycle starts, then guess a cycle on $(c_1, b_1)$, keeping only the current edge in memory.
	Now if there is a cycle on $(c_1, b_1)$, there has to be a cycle within $|C|$ length because within a cycle no counter value can change.
%	So this can be implemented in \EXPSPACE.
	Then we guess a path from $(c^{\prime\prime}, b^{c^{\prime\prime}})$ to $(c_1, b_1)$, and another path from $(c_1, b_1)$ to $(c_\tgt, b)$ (for some~$b$).
	The length of the part up to $(c_1, b_1)$ is bounded by
        %$|\bbC[\mu^k, c^{\prime\prime}]|$,
        $|C'|_r$,
        whilst we can always get a path of length at most $|C|$ (if a path exists) from $(c_1, b_1)$ to~$(c,b)$.
%        \SS*{This cannot be bounded by |C|?   [NM: this was an old comment I reinserted]}
%        \NM*{Why can the final part be bounded by $|C|$?}
	This discrepancy between two bounds is mainly because for the latter part of the path, we do not exactly fix the final vertex for the latter part ($b$~can be
        any tuple of $n$ nonnegative values),
%        anything with $\forall i \in \set{n},~b(i) \geq 0$),
        we just want the first component to be $c_\tgt$, while for the former part it gets fixed to $(c_1, b_1)$ when we guessed earlier.
	
	If we can guess such a path $\pi$ for each of $c^{\prime\prime} \in \dev_i(c,c^{\prime})$, we return $\mu^k_i(c, w,c^{\prime})=+\infty$.

\item 
  %Once we have $M \neq \infty$,
  otherwise,
  we have at least $1$ and at most $|V|$ many configurations $c^{\prime\prime} \in \dev_i(c,c^{\prime})$ such that $\sup_{\rho \in \LambdaSet_{\mu^{k-1}}(c^{\prime\prime})} \cost_i(\rho) \in \bbN$.
	We call this set of configurations as $\dev_i(c,c^{\prime})|_{\finsup}$

%	(1) There exists a $c^{\prime\prime} \in \dev_i(c, c^{\prime})|_{\finsup}$ such that for any valid path $\pi$ in $\LambdaSet_{\mu^{k-1}}(c^{\prime\prime})$, the corresponding path $\rho$ of $\calG$ satisfies $\cost_i(\rho) \leq M - cost_i(c,c^{\prime\prime})$, and length of $\pi$ is at most $|\bbC[\mu,c]|$.
%	Additionally, there need to be at least one valid path $\pi$ in $\LambdaSet_{\mu^{k-1}}(c^{\prime\prime})$ such that the corresponding path $\rho$ in $\calG$ satisfies $\cost_i(\rho) = M$.
%	(2) For all $c^{\prime\prime} \in \dev_i(c,c^{\prime})|{\finsup}$, there exists at least one $\pi$ in $\LambdaSet_{\mu^{k-1}}(c^{\prime\prime})$ such that, the corresponding path $\rho$ of $\calG$ satisfies $\cost_i(\rho) \geq M - \cost_i(c, c^{\prime\prime})$ and length of $\pi$ is bounded by $|\bbC[\mu^{k-1},c^{\prime\prime}]|$.
	
	To compute the finite value of~$\mu^{k}_i(c,w,c^{\prime})$,
        we compute $M_j = \sup_{\rho \in \LambdaSet_{\mu^{k-1}}(c^{\prime\prime}_j)} \cost_i(\rho)$ for each of $c^{\prime\prime}_j \in \dev_i(c, c^{\prime})|_{\finsup}$, and then take the minimum.
        In~order to compute~$M_j$, 
	we~first tentatively set $M_j = \min_{e \in E} \edgecost_e(1)$, and proceed iteratively as follows:
	we guess a valid path~$\pi$ of length at most $|C'|_r$ of $\bbC[\mu^{k-1}, c^{\prime\prime}_j]$ such that $\cost_i(\rho) \geq M_j$, where $\rho$ is the corresponding path of~$\pi$ in~$\calG$.
	If such a path exists, we increase~$M_j$ by~$1$, and repeat.
	At some point, we get $M_j$ such that there doesn't exist a valid path $\pi$ of length at most $|C'|_r$ with Player~$i$'s cost larger than or equal to~$M_j+1$, but there exists a valid path with Player~$i$'s cost larger than or equal to~$M_j$.
	We store that~$M_j$.
	
	When all values~$M_j$ have been computed, we return
        $M = \min_{j \in \set{|\dev_i(c,c^{\prime})|_{\finsup}|}}\{\cost_i(c,c^{\prime\prime}_j) + M_j\}$.
%	
%	We proceed iteratively as follows until $D$ equals $M$: we first guess a configuration $c^{\prime\prime} \in dev_i(c,c^{\prime})$, then guess a valid path $\pi$ in $\LambdaSet_{\mu^{k-1}}(c^{\prime\prime})$ such that $\cost_i(\rho) \geq D$, where $\rho$ is the corresponding path of $\pi$ in $\calG$ and $D = \min_{e \in E} \edgecost_e(1)$.
%	If such a path exists, we increase $D$ by $1$, and repeat.
%	If not, our guess about $M$ is wrong, so we take neq $M = D$ and start over.
%	At some point, there will exist a path $\rho$ such that $\cost_i(\rho) \geq D$, but there is no path $\pi$ in $\LambdaSet_{\mu^{k-1}}(c^{\prime\prime})$ for any such $c^{\prime\prime}$ such that Player $i$'s cost in the corresponding path is $\geq D+1$.
%	If $D$ equals the value of $M - cost_(c, c^{\prime\prime})$ for the $M$ we guessed, then we have guessed it right.
	We use at most doubly-exponential space in the procedure of guessing a path, and we reuse that space for guessing next path in the above algorithm.
\end{itemize} 

We keep another binary counter throughout transitioning from $\mu^{k-1}$ to $\mu^k$ to flag whether the fixpoint has been reached.
Once we reach $\lambda^* = \tuple{\lambda^*_i}$, we finally check whether $\LambdaSet_{\lambda^*}(c_\src)$ is empty by guessing a path in $\bbC[\lambda^*, c_\src]$ from $(c_\src, b^{c_\src})$ to $(c_\tgt, b)$.

In conclusion, we use doubly-exponential space:
(1)~to~store $\{\mu^{k-1}_i(\mathbf e)\mid i \in \set{n}, \mathbf e \in T\}$ which is double-exponential in the encoding of the number of players, and
(2)~to~encode a counter which keeps checking whether the length of our guessed paths does not exceed~$|C'|_r$.
%cross the corresponding $|\bbC[\mu^k, c]|$.\SS{Should we use $|C|_r$?}
%Hence,
\existsSPE*

%\NM*{\EXPSPACE requires space-efficient algorithm. Write a first version in 2\EXPTIME, we will deal with exact complexity at the end}

\end{document}